\documentclass[12pt]{amsart}
\usepackage{amsmath,amsthm,amsfonts,amssymb,eucal, hyperref}

\renewcommand {\a}{ \alpha }
\renewcommand{\b}{\beta}

\newcommand{\g}{\gamma}

\renewcommand{\L}{\Lambda}

\renewcommand{\t}{\theta}

\newcommand{\om}{\omega}
\newcommand{\Om}{\Omega}

\newcommand{\oq}{\ {\raise 7pt\hbox{${\scriptstyle\circ}$}}
\kern -7pt{
\hbox{$Q$}}}

\newcommand{\R}{ \mathbb R}
\newcommand{\N}{ \mathbb N}

\newcommand{\Rd}{ \mathbb R^d}

\newcommand {\GF}{\mathfrak F}
\newcommand {\GH}{\mathfrak H}

\newcommand {\GV}{\mathfrak V}
\newcommand {\GW}{\mathfrak W}
\newcommand {\GU}{\mathfrak U}

\newcommand {\GX}{\mathfrak X}
\newcommand {\ba}{\mathbf a}

\newcommand {\BD}{\mathbf D}

\newcommand {\BM}{\mathbf M}

\newcommand {\BS}{\mathbf S}

\newcommand {\bx}{\mathbf x}

\newcommand {\be}{\mathbf e}

\newcommand {\bk}{\mathbf k}

\newcommand {\by}{\mathbf y}
\newcommand {\bt}{\mathbf t}

\newcommand {\bs}{\mathbf s}

\newcommand {\bv}{\mathbf v}
\newcommand {\bn}{\mathbf n}

\newcommand {\bnu}{\boldsymbol\nu}

\newcommand {\bmu}{\boldsymbol\mu}
\newcommand {\bth}{\boldsymbol\theta}
\newcommand {\Bth}{\boldsymbol\Theta}
\newcommand {\boldeta}{\boldsymbol\eta}
\newcommand {\boldom}{\boldsymbol\om}
\newcommand {\bphi}{\boldsymbol\phi}
\newcommand {\bxi}{\boldsymbol\xi}

\newcommand {\BXi}{\boldsymbol\Xi}
\newcommand {\Bxi}{\boldsymbol\Xi}

\newcommand{\bchi}{\boldsymbol\chi}

\newcommand{\lu}{\langle}
\newcommand{\ru}{\rangle}

\newcommand{\CV}{\mathcal V}

\newcommand{\CX}{\mathcal X}

\newcommand{\CP}{\mathcal P}

\newcommand{\CA}{\mathcal A}

\newcommand{\CC}{\mathcal C}

\newcommand{\plainC}[1]{\textup{{\textsf{C}}}^{#1}}

\newcommand{\plainS}{\textup{{\textsf{S}}}}

\newcommand{\1}
{{\,\vrule depth3pt height9pt}{\vrule depth3pt height9pt}
{\vrule depth3pt height9pt}{\vrule depth3pt height9pt}\,}

\DeclareMathOperator{\op}{{Op}}

\newtheorem{thm}{Theorem}[section]
\newtheorem{cor}[thm]{Corollary}
\newtheorem{cla}[thm]{Claim}
\newtheorem{lem}[thm]{Lemma}
\newtheorem{prop}[thm]{Proposition}

\theoremstyle{definition}
\newtheorem{defn}[thm]{Definition}

\newtheorem*{remark}{Remark}
\newtheorem{rem}[thm]{Remark}

\numberwithin{equation}{section}

\newcommand{\bee}{\begin{equation}}
\newcommand{\ene}{\end{equation}}
\newcommand{\bees}{\begin{equation*}}
\newcommand{\enes}{\end{equation*}}
\newcommand{\bes}{\begin{split}}
\newcommand{\ens}{\end{split}}

\newcommand{\bet}{\begin{thm}}
\newcommand{\ent}{\end{thm}}
\newcommand{\bel}{\begin{lem}}
\newcommand{\enl}{\end{lem}}
\newcommand{\bec}{\begin{cor}}
\newcommand{\enc}{\end{cor}}
\newcommand{\becl}{\begin{cla}}
\newcommand{\encl}{\end{cla}}
\newcommand{\bep}{\begin{proof}}
\newcommand{\enp}{\end{proof}}
\newcommand{\ber}{\begin{rem}}
\newcommand{\enr}{\end{rem}}
\newcommand{\ep}{\varepsilon}
\newcommand{\la}{\lambda}
\newcommand{\La}{\Lambda}
\newcommand{\de}{\delta}
\newcommand{\al}{\alpha}
\newcommand{\Z}{\mathbb Z}

\newcommand {\BUps}{\boldsymbol\Upsilon}

\newcommand{\De}{\Delta}
\newcommand{\CF}{\mathcal F}

\makeatletter
\def\square{\RIfM@\bgroup\else$\bgroup\aftergroup$\fi
  \vcenter{\hrule\hbox{\vrule\@height.6em\kern.6em\vrule}\hrule}\egroup}
\makeatother

\begin{document}

\hoffset -4pc

\title[Local Density of states]
{Complete asymptotic expansion of the spectral function of multidimensional almost-periodic Schr\"odinger operators}
\author[L. Parnovski \& R. Shterenberg]
{Leonid Parnovski \& Roman Shterenberg}
\address{Department of Mathematics\\ University College London\\
Gower Street\\ London\\ WC1E 6BT\\ UK}
\email{Leonid@math.ucl.ac.uk}
\address{Department of Mathematics\\ University of Alabama at Birmingham\\ 1300 University Blvd.\\
Birmingham AL 35294\\ USA}
\email{shterenb@math.uab.edu}


\date{\today}


\begin{abstract}
We prove the complete asymptotic expansion of the spectral function (the integral kernel of the spectral projection) of a
Schr\"odinger operator $H=-\Delta+b$ acting in $\R^d$ when the potential $b$ is real and either
smooth periodic, or generic quasi-periodic (finite linear combination of exponentials), 
or belongs to a wide class of almost-periodic functions. 
\end{abstract}

\



\maketitle
\vskip 0.5cm

\renewcommand{\uparrow}{{\mathcal {L F}}}
\newcommand{\ssharp}{{\mathcal {L E}}}
\renewcommand{\natural}{{\mathcal {NR}}}
\renewcommand{\flat}{{\mathcal R}}
\renewcommand{\downarrow}{{\mathcal {S E}}}

\section{Introduction}

We consider the Schr\"odinger operator
\bee\label{eq:Sch}
H=-\Delta+b
\ene
acting in $\R^d$. The potential $b=b(\bx)$ is assumed to be real, smooth, and either periodic, or almost-periodic;
in the almost-periodic case we assume that all the derivatives of $b$ are almost-periodic as well. Let $E_{\la}=E_{\la}(H)$ be the spectral projection of $H$ and $e_{\la}(\bx,\by)=e_{\la}(H;\bx,\by)=e(\la;H;\bx,\by)$  
be its integral kernel (also called the {\it spectral function}). We put $N(\la;\bx)=N(\la;\bx;H):=e_{\la}(\bx,\bx)$ and call $N(\la;\bx)$ the {\it Local Density of States} (LDS) of $H$. The study of the asymptotic behaviour of the LDS (for much more general classes of operators) has been the subject of many papers, see e.g. \cite{AgKa,Ca,Gar,Ho,Le,Sa,Si}.

The Pastur-Shubin theorem implies that the integrated density of states (IDS) $N(\la;H)$ is the mean of the local density of states over the spatial variable:
\bee
N(\la)={\mathbf{M}}_{\bx}N(\la;\bx)=
\lim_{\L\to\infty}\frac{\int_{[-\L,\L]^d}N(\la;\bx)d\bx}{(2\L)^d}.
\ene
In our paper \cite{PaSh1}, we have proved that, subject to several assumptions, the IDS admits a complete asymptotic expansion:
\bee\label{eq:intr0}
N(\lambda)\sim \la^{d/2}\Bigl(C_d+\sum_{j=1}^\infty a_j\lambda^{-j}\Bigr),
\ene
meaning that for each $L\in\N$ one has
\bee\label{eq:intr1}
N(\lambda)=\la^{d/2}\Bigl(C_d+\sum_{j=1}^L a_j\lambda^{-j}\Bigr)+R_L(\lambda)
\ene
with $R_L(\lambda)=o(\la^{\frac{d}{2}-L})$. In those formulas, 
\bee
C_d=\frac{w_d}{(2\pi)^d} \ \text{and} \ w_d=\frac{\pi^{d/2}}{\Gamma(1+d/2)}
\ene
is a volume of the unit ball in $\R^d$; coefficients
$a_j$ are real numbers that depend on the potential $b$.
They can be calculated relatively easily using the heat kernel invariants (computed in \cite{HitPol}) and the results of \cite{KorPush}; they are equal to certain
integrals of the potential $b$ and its derivatives. Below, we give more details about the history of proving \eqref{eq:intr0}.

The first aim of our paper is to prove the `localised' version of \eqref{eq:intr0}:
\bee\label{eq:intrlocal}
N(\lambda;\bx)\sim \la^{d/2}\Bigl(C_d+\sum_{j=1}^\infty a_j(\bx)\lambda^{-j}\Bigr).
\ene 
According to \cite{HitPol} and \cite{KorPush}, if \eqref{eq:intrlocal} holds, we must have
\bee
a_j(\bx)=\frac{\sigma_j(\bx)}{(4\pi)^{d/2}\Gamma(\frac{d}{2}-j+1)},
\ene
where $\sigma_j(\bx)$ are local heat invariants given by
$$
\sigma_j(\bx)=\sum\limits_{k=0}^j\frac{(-1)^j\Gamma(j+d/2)}{4^k k! (k+j)! (j-k)!\Gamma(k+d/2+1)}H_\by^{k+j}(|\bx-\by|^{2k})\big|_{\by=\bx}.
$$
Here, $H_\by$ is our operator \eqref{eq:Sch} acting in variable $\by$. 
Moreover, we obviously have
\bee
a_j={\mathbf{M}}_{\bx}a_j(\bx). 
\ene
It is clear that \eqref{eq:intrlocal} (with remainder estimates being uniform in $\bx$) immediately implies \eqref{eq:intr0}, but the opposite is not true. Formula \eqref{eq:intrlocal} has been proved in the one-dimensional periodic case in \cite{ShuSche}. 

\ber
Suppose for the moment that $b$ belongs to a bigger class of potentials: $b$ is bounded together with all its derivatives; we denote the collection of all such potentials by $\mathrm{USB}(\R^d)$ (this stands for \emph{uniformly smoothly bounded}). Then the IDS of $H$ may be not well-defined, but the LDS still exists. We formulate two conjectures about the asymptotic behaviour of LDS in this wider class of operators. 

{\bf  Conjecture 1:} Asymptotic formula \eqref{eq:intrlocal} holds for Schr\"odinger operators with USB potentials.
 
{\bf  Conjecture 2:} Suppose, two potentials $b_1,b_2\in \mathrm{USB}(\R^d)$ coincide in a neighbourhood of $\bx$. Then 
\bee
N(\la,\bx;H_1)-N(\la,\bx;H_2)=O(\la^{-\infty}).
\ene

As \cite{HitPol} and \cite{PoSh} show, Conjecture 1 and Conjecture 2 are equivalent. To the best of our knowledge, the only situation (apart from the periodic and almost-periodic cases established in our paper) when Conjectures 1 and 2 have been proved is when $b$ has compact support, \cite{PoSh} and \cite{Va}. Unfortunately, it does not look likely that the method used in these papers can be extended to a bigger class of operators. Indeed, this method, if it works, allows one to obtain the complete asymptotic expansion not just for the spectral function $e$, but even for its derivative with respect to $\lambda$. Obviously, such an expansion cannot exist for an arbitrary  $\mathrm{USB}$ potential; 
most periodic potentials in $\R$ give the obvious counter-examples (because of the existence of infinitely many spectral gaps). 
It seems that Conjectures 1-2 are not known in full generality even in the one-dimensional case.
\enr

The second result of this paper is obtaining the information about the asymptotic behaviour of the spectral function off the diagonal. In the off-diagonal case we obtain the complete asymptotic expansion for the so called non-degenerate directions $\frac{\bx-\by}{|\bx-\by|}$ which form the set of full measure on the unit $(d-1)$-dimensional sphere (see Theorem \ref{th:main2} for the exact formulation); for such directions we prove that 
\bee\label{eq:int_main_thm2}
\begin{split}
e_{\la}(\bx,\by)&\sim 
\cos(\lambda^{1/2}|\bx-\by|)\sum\limits_{q=0}^{\infty}\acute a_q(\bx,\by)\lambda^{(d-1)/4-q/2}\cr & +\sin(\lambda^{1/2}|\bx-\by|)\sum\limits_{q=0}^{\infty}\grave a_q(\bx,\by)\lambda^{(d-1)/4-q/2}.
\end{split}
\ene
More precisely, in all cases (on and off-diagonal, both degenerate and non-degenerate) we have reduced the problem of finding the asymptotic expansion of the spectral function to computing certain rather complicated integrals. In the diagonal case these integrals can be computed with the brute force, whereas in the non-degenerate off-diagonal case these integrals can be computed (or rather approximated) using the stationary phase method. Computing these integrals in the degenerate off-diagonal case is technically too difficult a task; we may return to it in a further publication. In the one-dimensional periodic case \eqref{eq:int_main_thm2} was also obtained in \cite{ShuSche}. Unfortunately, unlike in the diagonal case, we cannot say 
much about the coefficients $\acute a_q(\bx,\by)$, $\grave a_q(\bx,\by)$, since the results of \cite{HitPol} are not known off the diagonal. The results of \cite{Mol} show that 
these coefficients depend only on the behaviour of the potential in the neighbourhood of the interval joining $\bx$ and $\by$.

There is a long history of results related to proving the expansions \eqref{eq:intr0}-\eqref{eq:intr1} for the IDS. In the one-dimensional case the complete expansion \eqref{eq:intr0} was obtained in \cite{ShuSche} for periodic potentials, and in \cite{Sav} in the almost-periodic case. 
For higher dimensional periodic operators, the important steps were:  \cite{HelMoh}, \cite{Kar1},
\cite{Kar}, \cite{ParSob}, \cite{Skr}, \cite{Sob1}. At last, the complete expansion \eqref{eq:intr0} was obtained in \cite{ParSht} for $d=2$ and in \cite{PaSh1} for arbitrary $d$. 
Finally, in the multidimensional almost-periodic case,
formula \eqref{eq:intr1} was known only with $L=0$ and $R(\la)=O(\la^{\frac{d-2}{2}})$, see \cite{Shu}, until \eqref{eq:intr0}  
was proved in \cite{PaSh1}. 

On the other hand, if we talk about the asymptotic expansions \eqref{eq:intrlocal} and \eqref{eq:int_main_thm2} of the spectral function, then, with the exception of the already mentioned paper \cite{ShuSche} where these formulas were obtained in the case $d=1$ and $b$ periodic, the only other results, to the best of our knowledge, were one-term asymptotics of the LDS and zero-term asymptotics (i.e., optimal estimates, without the first term) off the diagonal. 

Now let us discuss the method we employ to prove \eqref{eq:intrlocal} and \eqref{eq:int_main_thm2} and the additional difficulties we have encountered compared with the proof of \eqref{eq:intr0}. 
Let us assume that $\la$ belongs to a spectral interval $[\la_0,2\la_0]$ and obtain the asymptotic expansion there; it is a relatively simple task (explained in Section 3) how to `glue' asymptotic expansions obtained in different spectral intervals.  

The first step, as  in \cite{PaSh1}, is to perform the gauge transform to $H$ (see also \cite{Sob,Sob1,ParSob}). This results in obtaining two operators, $H_1$ and $H_2$ such that: 

1. $H_1$ is unitary equivalent to $H$: $H_1=UHU^{-1}$, with an explicit (although complicated) formula for $U$;

2. $H_1$ and $H_2$ are close to each other: $||H_2-H_1||\le\la^{-N}$, where $N$ is arbitrarily large (but fixed) number. In fact, we will need even better `closeness': we will show that $||(H_2-H_1)(-\De+I)^s||\le\la^{-N}$, where both $s$ and $N$ are large (but fixed). 

3. Finally, $H_2$ is `almost diagonal' in the interval $[\la_0,2\la_0]$. 
This means that for a large portion of values of the dual variable $\bxi$ inside the annulus $\{\bxi\in\R^d, \la_0\le|\bxi|^2\le 2\la_0\}$, the symbol of $H_2$ has constant coefficients (so $H_2$ has no off-diagonal terms for such $\bxi$). These `good' values of $\bxi$ belong to the so called {\it non-resonant region}. In the other, so called {\it resonant regions}, many (but not all) off-diagonal terms of $H_2$ are also zeros. As a result, $H_2$ has many invariant subspaces. 

The next step is to compare the spectral functions of $H$ and $H_2$. Let us recall how this step was done when we were studying the IDS in \cite{PaSh1}. First, we proved that $N(\la;H)=N(\la;H_1)$ using the representation of the IDS as the von Neumann trace of $E_{\la}$ and the basic properties of this trace. Then we have used the fact that the IDS is monotone with respect to the operator and therefore, if $||H_2-H_1||\le\la^{-N}$, this implies that 
\bee
N(\la-\la^{-N};H_2)\le N(\la;H_1)\le  N(\la+\la^{-N};H_2). 
\ene

The first step (from $H$ to $H_1$) is rather simpler when we study the spectral function: we use the fact that (at least formally) we have  
\bee\label{eintr}
\bes
&e(\la;H;\bx_0,\by_0)=(E_{\la}(H)\de_{\bx_0},\de_{\by_0})\\
&=
(U^{-1}E_{\la}(H_1)U\de_{\bx_0},\de_{\by_0})=(E_{\la}(H_1)U\de_{\bx_0},U\de_{\by_0}).
\end{split}
\ene
It is a much more serious problem to switch from $H_1$ to $H_2$. In general, it is obviously not true that if we change the operator by something small, then the spectral projection is changed by something small. This statement, however, becomes true if we consider the change of the spectral projections in a certain weak sense, see Lemma \ref{lem:2}. This Lemma is probably the first of two important new ideas 
in our paper. Lemma \ref{L2norm} then shows that despite the fact that the delta-function does not belong to $L_2$, the function $E_{\la}\de_{\bx_0}$ is inside $L_2$ (with the control of its norm) which makes legal most of the formal computations. 

The next step is to compute 
\bee\label{eq:17}
(E_{\la}(H_2)U\de_{\bx_0},U\de_{\by_0}).
\ene
Here, we use the trick ideologically similar to formula (10.18) in \cite{PaSh1} when, in order to calculate a certain object for a real analytic family of operators, we extend this family to the complex plane,  express this object as a contour integral and then, after a chain of manipulations (expanding our integral in geometric series and using the Cauchy integral formula), we return back to the real axis having expressed the difficult object in a convenient and explicit form. This time we need to express the spectral projection of a real analytic family of operators. We again go into the complex plane, write the spectral projection as the Riesz integral and then change the variables so that instead of integrating against the spectral parameter, we are integrating against the parameter of the family. Afterwards, we similarly return back to the real axis and express \eqref{eq:17} in the explicit form. This is done in  
Lemmas \ref{aux} and \ref{spectralint} and formula \eqref{eq:Cauchy}. This is the second important new idea of our paper. 

After these steps, we have reduced the problem to computing certain explicit (though complicated) integrals. In the diagonal case, these integrals are precisely of the form that was already computed (essentially by brute force) in \cite{PaSh1}. Off the diagonal, the integrals become too complicated to compute by hand, but we can use instead the stationary phase method to compute them. This is where we use the fact that the direction 
$\bx-\by$ is non-degenerate: otherwise, even the stationary phase integrals become too involved. 

As it has already been mentioned, many constructions and results needed for our proof are either identical, or similar to corresponding statements from \cite{PaSh1}. This refers, in particular, to most of Sections 5 and 6. In order to keep the size of our paper reasonable, but make it self-contained, we have been using the following convention: we write in detail all the definitions and statements from \cite{PaSh1} necessary for our proof. If the proof of a certain statement is identical (or essentially identical) to the proof of the corresponding statement of \cite{PaSh1}, we omit it. However, if the proof requires substantial changes (like, e.g., the proof of Lemma \ref{th:main2withconstant}), we write it here completely. Still, we believe that it would help to understand our paper better if the reader reads \cite{PaSh1} first.

Another convenient convention that we were already using in \cite{PaSh1} is this. Let $A$ be an  elliptic pseudo-differential operator
with almost-periodic coefficients. Usually, we are assuming that $A$ acts in $L_2(\R^d)$. However, we can consider actions of $A$ (via the same
Fourier integral operator formula) in the Besicovitch space $B_2(\R^d)$. The space $B_2(\R^d)$ is the space of all formal
sums
\bees
\sum_{j=1}^\infty c_{j}\be_{\bth_j}(\bx),
\enes
where
\bee
\be_{\bth}(\bx):=e^{i\langle\bth,\bx\rangle}
\ene
and $\sum_{j=1}^\infty |c_{j}|^2<+\infty$. It is known (see \cite{Shu0}) that the spectra of $A$ acting in $L_2(\R^d)$ and $B_2(\R^d)$ are the same, although the types of
those spectra can be entirely different. It is very convenient, when working with the gauge transform constructions, to assume that all the operators involved
act in $B_2(\R^d)$, although in the end we will return to operators acting in $L_2(\R^d)$. This trick (working with operators acting in $B_2(\R^d)$)
is similar to working with fibre operators $A(\bk)$ in the periodic case in a sense that we can freely consider the action of an operator on one, or finitely many,
exponentials, without caring that these exponentials do not belong to our function space. In most of the situations it will be clear from the context which is the space we work in, but sometimes we will indicate this by writing  $A^L$ (resp. $A^B$) for actions in $L_2(\R^d)$ (resp. in $B_2(\R^d)$). 

During our computations, we will obtain some `extra' asymptotic terms that are absent in the final expansion (compare e.g. \eqref{eq:main_thm1} with \eqref{eq:main_cor1}). The way we get rid of these extra terms is different in the on and off-diagonal cases. On the diagonal we use the a priori form of the asymptotic expansion given by the asymptotics of the heat kernel computed in \cite{HitPol} and \cite{HitPol1}. 
Off the diagonal, we use the Seeley type formula for the meromorphic extension of the complex powers of $H$ (although we could have used 
the heat kernel extension obtained in \cite{Mol} and \cite{Kan}).

The plan of the paper is as follows. In Section 2, we give necessary definitions and formulate the main results. In Section 3, we discuss how to `glue' asymptotic expansions obtained in different intervals of the spectral parameter and how to get rid of the `extra' asymptotic terms. In Section 4, we prove several auxiliary statements (since these statements are quite crucial for our method, we have decided to prove them in a special Section rather than to move their proofs to an Appendix). In Section 5, we introduce the resonance regions and the coordinates in these 
regions. In Section 6, we discuss the classes of pseudo-differential operators we will work in and introduce the method of the gauge transform. Finally, in Section 7 we finish the proofs of the main statements.

\subsection*{Acknowledgments}
We are grateful to Iosif Polterovich for useful discussions. 
The research of the first author was partially supported by the EPSRC grant EP/J016829/1.

\section{Notation and Main Results}

Since our potential $b$ is almost-periodic, it has the Fourier series
\bee\label{eq:potential}
b(\bx)\sim\sum_{\bth\in\Bth}\hat b({\bth})\be_{\bth}(\bx),
\ene
where $\Bth$ is a (countable) set of frequencies. 
\ber
Although for general almost-periodic functions the series \eqref{eq:potential} does not need to be convergent, the assumptions we impose on $b$ later will imply that \eqref{eq:potential} is, in fact, an equality. If $b$ is periodic, then $\Bth\subset\Gamma^{\dagger}$, where
$\Gamma^{\dagger}$ is the lattice dual to the lattice $\Gamma$ of periods of $b$. 
\enr
Without loss of generality we assume that $\Bth$ spans $\R^d$ and contains $0$; we also put 
\bee\label{eq:algebraicsum}
\Bth_k:=\Bth+\Bth+\dots+\Bth 
\ene
(algebraic sum taken $k$ times) and $\Bth_{\infty}:=\cup_k\Bth_k=Z(\Bth)$, where for a set $S\subset \R^d$ by $Z(S)$ we denote the set of all finite linear combinations of elements in
$S$ with integer coefficients. The set $\Bth_\infty$ is countable and non-discrete (unless $b$ is periodic, in which case $\Bth_\infty=\Gamma^{\dagger}$).
The first condition we impose on the potential is:

{\bf Condition A}. Suppose that $\bth_1,\dots,\bth_d\in \Bth_\infty$. Then $Z(\bth_1,\dots,\bth_d)$ is discrete.

It is easy to see that this condition can be reformulated like this:
suppose, $\bth_1,\dots,\bth_d\in \Bth_\infty$.
Then
either $\{\bth_j\}$ are linearly independent, or $\sum_{j=1}^d n_j\bth_j=0$, where $n_j\in\Z$ and not all
$n_j$ are zeros. This reformulation shows that Condition A is generic: indeed, if we are choosing frequencies
of $b$ one after the other, then on each step we have to avoid choosing a new frequency from a countable set of
hyperplanes, and this is obviously a generic restriction. Yet another equivalent reformulation of this condition is as follows: Let $\GV$ be any proper linear subspace of $\R^d$. Denote
\bee\label{eq:potentialGV}
b_{\GV}(\bx):=\sum_{\bth\in\Bth\cap\GV}\hat b({\bth})\be_{\bth}(\bx).
\ene
Then $b_{\GV}$ is periodic. 
Condition A is clearly always satisfied for periodic
potentials, but it becomes meaningful for quasi-periodic potentials. If $d=2$, this condition simply means that any two collinear frequencies are commensurate. 

Our main result will hold in the cases when $b$ is smooth and either periodic or quasi-periodic satisfying condition A. The rest of the conditions on the potential are required only when it is `truly' almost-periodic. 
These extra conditions state that we have a tight control over the approximations of $b$ by quasi-periodic functions. 
In the proof we are going to work
with quasi-periodic approximations of $b$, and we need these conditions to make sure that
all estimates in the proof are uniform with respect to these approximations.

{\bf Condition B}.
Let $k$ be an arbitrary fixed natural number. Then for each sufficiently large real number $\rho$
there is a finite set $\Bth(k;\rho)\subset(\Bth\cap B(\rho^{1/k}))$ (where $B(r)$ is a ball of radius $r$ centered at
$0$) and a `cut-off' potential
\bee\label{eq:condB1}
b_{(k;\rho)}(\bx):=\sum_{\bth\in\Bth(k;\rho)}\hat b'({\bth})\be_{\bth}(\bx)
\ene
which satisfies
\bee\label{eq:condB2}
||b-b_{(k;\rho)}||_{\infty}<\rho^{-k}.
\ene 

The next condition we need to impose is a version of the Diophantine condition on the frequencies of $b$. First, we need
some definitions. We fix a natural number $\tilde k$ (the choice of $\tilde k$ will be determined later by how many terms in \eqref{eq:intr1}
we want to obtain) and denote $\tilde\Bth:=[\Bth(k;\rho)]_{\tilde k}$ 
(see \eqref{eq:algebraicsum} for the notation) and 
$\tilde\Bth':=\tilde\Bth\setminus\{0\}$.
We say that $\GV$ is a quasi-lattice subspace of dimension $m$, if $\GV$ is a linear
span of $m$ linear independent vectors $\bth_1,\dots,\bth_m$ with $\bth_j\in\tilde\Bth\ \forall j$. Obviously, zero
space (which we will denote by $\GX$)
is a quasi-lattice subspace of dimension $0$ and $\R^d$ is a quasi-lattice subspace of dimension $d$.
We denote by $\CV_m$ the collection of all quasi-lattice subspaces of dimension $m$ and put
$\CV:=\cup_m\CV_m$.
If $\bxi\in\R^d$
and $\GV$ is a linear subspace of $\R^d$, we denote by $\bxi_{\GV}$ the orthogonal projection of $\bxi$
onto $\GV$, and put $\GV^\perp$ to be an orthogonal complement of $\GV$, so that $\bxi_{\GV^\perp}=\bxi-\bxi_{\GV}$.
Let $\GV,\GU\in\CV$. We say that these subspaces are {\it strongly distinct}, if neither of them is a
subspace of the other one. This condition is equivalent to stating that if we put $\GW:=\GV\cap\GU$, then
$\dim \GW$ is strictly less than dimensions of $\GV$ and $\GU$. We put $\phi=\phi(\GV,\GU)\in [0,\pi/2]$
to be the angle between them, i.e. the angle between $\GV\ominus\GW$ and $\GU\ominus\GW$, where $\GV\ominus\GW$
is the orthogonal complement of $\GW$ in $\GV$. This angle is positive iff $\GV$ and $\GW$ are strongly distinct.
We put $s=s(\rho)=s(\tilde\Bth):=\inf\sin(\phi(\GV,\GU))$, where infimum is over all strongly distinct pairs of subspaces from $\CV$,
$R=R(\rho):=\sup_{\bth\in\tilde\Bth}|\bth|$, and $r=r(\rho):=\inf_{\bth\in\tilde{\Bth}'}|\bth|$. Obviously, $R(\rho)\ll \rho^{1/k}$ (where the
implied constant can depend on $k$ and $\tilde k$).

{\bf Condition C}. For each fixed $k$ and $\tilde k$ the sets $\Bth(k;\rho)$ satisfying \eqref{eq:condB1} and \eqref{eq:condB2}
can be chosen in such a way that for sufficiently large $\rho$ we have
\bee\label{eq:condC1}
s(\rho)\ge\rho^{-1/k}
\ene
and
\bee\label{eq:condC2}
r(\rho)\ge\rho^{-1/k},
\ene
where the implied constant (i.e. how large should $\rho$ be) can depend on $k$ and $\tilde k$.

\ber
One can understand Conditions B and C in the following way. These conditions specify how quickly the Fourier coefficients of $b$ should decay, given the Diophantine properties of the frequencies. 
\enr

Now we can formulate our first theorem.

\bet\label{th:main1}
Let $H$ be an operator \eqref{eq:Sch} 
with smooth real potential $b$ which is either periodic, or quasi-periodic satisfying Condition {\rm A}, or almost-periodic 
satisfying Conditions {\rm A,B,} and {\rm C}.
Then for each $L\in\N$ we have (uniformly in $\bx\in\R^d$):
\bee\label{eq:main_thm0}
N(\la;\bx)=\la^{d/2}\left(C_d+\sum\limits_{j=1}^{L}a_j(\bx)\la^{-j}+o(\la^{-L})\right)
\ene
as $\la\to\infty$.
\ent
\ber
Following \cite{HitPol}, \cite{HitPol1}, and \cite{KorPush}, it is straightforward to compute the coefficients $a_j$. For
example, we have
\bees
a_1(\bx)=-\frac{d w_d}{2(2\pi)^d}b(\bx)
\enes
and
\bees
a_2(\bx)=
\frac{d(d-2) w_d}{24(2\pi)^d}(3b^2(\bx)-\Delta\,b(\bx)).
\enes
\enr
Our second result concerns the off-diagonal behaviour of the spectral function. 
\bet\label{th:main2}
Let $H$ be an operator satisfying all the conditions of the previous Theorem. Suppose that the direction $\frac{\bx-\by}{|\bx-\by|}$ is not orthogonal to any of the vectors in $\Bth_\infty\setminus\{{\bf 0}\}$. Then for each $L\in\N$ we have:
\bee\label{eq:main_thm2}
\begin{split}
e_{\la}(\bx,\by)&=
\cos(\lambda^{1/2}|\bx-\by|)\sum\limits_{q=0}^{2L}\acute a_q(\bx,\by)\lambda^{(d-1)/4-q/2}\cr & +\sin(\lambda^{1/2}|\bx-\by|)\sum\limits_{q=0}^{2L}\grave a_q(\bx,\by)\lambda^{(d-1)/4-q/2}+o(\lambda^{(d-1)/4-L}),
\end{split}
\ene
as $\la\to\infty$. The asymptotic expansion is uniform along every non-degenerate  direction when $|\bx-\by|$ is bounded and separated away from $0$. 
\ent

\begin{remark} 1. Obviously, the coefficients $\acute a_q(\bx,\by),\ \grave a_q(\bx,\by)$ are real-valued  but unlike the on-diagonal case, here we don't know them explicitly. It is nevertheless possible to compute first few coefficients using our constructions. In particular, $\acute a_0$ and $\grave a_0$ are (as expected) the same as for the free operator $-\Delta$:
\bee\label{free}
e_\la(\bx,\by)=\frac{2}{(2\pi|\bx-\by|)^{(d+1)/2}}\la^{(d-1)/4}\sin\left(\la^{1/2}|\bx-\by|-\frac{\pi(d-1)}{4}\right)(1+O(\la^{-1/2})).
\ene
Moreover, \eqref{free} holds for all $\bx\not=\by$ including degenerate directions.

2. Our set of non-degenerate directions has full measure but in general is not open. Still, as can be seen from the proof, for every fixed $L$ formula \eqref{eq:main_thm2} holds for all directions not orthogonal to any of the vectors in $\Bth_{\tilde k}\setminus\{{\bf 0}\}$, $\tilde k=\tilde k(L)$, the latter set being just finite. Corresponding partial expansion is uniform in any compact set within these directions and outside of $|\bx-\by|=0$, coefficients $\acute a_q(\bx,\by),\ \grave a_q(\bx,\by)$ being smooth.


\end{remark}

As we have mentioned earlier, certain parts of the proof are virtually identical to corresponding parts of \cite{PaSh1} and will be omitted. In particular, at the end of Section 3 of \cite{PaSh1} it is explained how to obtain the asymptotic formula for the IDS in the almost-periodic situation assuming we can obtain it for quasi-periodic potentials. This explanation works in the case of LDS (and the spectral function off the diagonal) as well. Therefore, we will prove our results only for quasi-periodic potentials  and from now on we assume that $b$ has finitely many frequencies and, thus, that $\Bth(k;\rho)=\Bth$.

In this paper, by
$C$ or $c$ we denote positive constants,
The exact value of which can be different each time they
occur in the text,
possibly even each time they occur in the same formula. On the other hand, the constants which are labeled (like $C_1$, $c_3$, etc)
have their values being fixed throughout the text.
Given two positive functions $f$ and $g$,
we say that $f\gg g$, or $g\ll
f$, or $g=O(f)$ if the ratio $\frac{g}{f}$ is bounded. We say
$f\asymp g$ if $f\gg g$ and $f\ll g$.

\section{More notation and auxiliary results}

In this section, we start explaining our method. Let us put $\rho:=\sqrt{\la}$. The first result of our paper (about the LDS) is a consequence of the following theorem:

\bet\label{main_thm}
For each $L\in\N$ we have (uniformly in $\bx\in\R^d$):
\bee\label{eq:main_thm1}
N(\rho^2;\bx)=C_d\rho^d+\sum_{p=0}^{d}\sum_{j=-d+1}^{L}a_{j,p}(\bx)\rho^{-j}(\ln\rho)^p
+o(\rho^{-L})
\ene
as $\rho\to\infty$. 
\ent
Once the theorem is proved, it immediately implies
\bec
For each $L\in\N$ we have (uniformly in $\bx\in\R^d$):
\bee\label{eq:main_cor1}
N(\la;\bx)=\la^{d/2}\left(C_d+\sum\limits_{j=1}^{L}a_j(\bx)\la^{-j}+o(\la^{-L})\right)
\ene
as $\la\to\infty$.
\enc
\bep
The proof is the same as the proof of Corollary 3.2 from \cite{PaSh1}. 
\enp

Next we choose sufficiently large $\rho_0>1$ (to be fixed later on) and put $\rho_n=2\rho_{n-1}=2^n\rho_0$,
$\la_n:=\rho_n^2$;
we also define the interval
$I_n=[\rho_n,4\rho_n]$.
The proof of Theorem \ref{main_thm} will be based on the following lemma:

\bel\label{main_lem}
For each $M\in\N$ and $\rho\in I_n$ we have:
\bee\label{eq:main_lem1}
N(\rho^2;\bx)=C_d\rho^d+\sum_{p=0}^{d}
\sum_{j=-d+1}^{6M}a_{j,p;n}(\bx)\rho^{-j}(\ln\rho)^p
+O(\rho_n^{-M}).
\ene
Here, $a_{j,p;n}(\cdot):\R^d\to\R$ are some functions depending on $j,\ p$ and $n$ (and $M$) satisfying
\bee\label{eq:main_lem2}
a_{j,p;n}(\bx)=O(\rho_n^{(2j/3)+d+1}).
\ene
The constants in the $O$-terms do not depend on $n$ or $\bx$ (but they may depend on $M$). 
\enl
\ber\label{rem:new1}
Note that \eqref{eq:main_lem1} is not a `proper' asymptotic formula, since the coefficients
$a_{j,p;n}(\bx)$ are allowed to grow with $n$ (and, therefore, with $\rho$).
\enr
In Section 3 of \cite{PaSh1} it is explained, how to prove 
Theorem \ref{main_thm} assuming that Lemma \ref{main_lem} is established (see also the details for the off-diagonal case below). Therefore, what we have to do is to prove Lemma \ref{main_lem}.   

For the off-diagonal case the technical result which we prove is the following.
\bel\label{main_lemoff}
Let the direction $\frac{\bx-\by}{|\bx-\by|}$ be not orthogonal to any of the vectors in $\Bth_\infty\setminus\{{\bf 0}\}$. Then  for each $M\in\N$ and $\rho\in I_n$ we have:
\begin{equation}\label{eq:main_lemoff}
\begin{split}&
e_{\rho^2}(\bx,\by)=
\cos(\rho|\bx-\by|)\sum\limits_{p=-d+1}^{4M}\hat a(p;n)(\bx,\by)\rho^{-p-(d-1)/2}+\cr & \sin(\rho|\bx-\by|)\sum\limits_{p=-d+1}^{4M}\check a(p;n)(\bx,\by)\rho^{-p-(d-1)/2}+A_0(n)(\bx,\by)+O(\rho_n^{-M}).
\end{split}
\end{equation}
Here, $\hat a(p;n)(\cdot,\cdot),\ \check a(p;n)(\cdot,\cdot),\ A_0(n)(\cdot,\cdot):\,\R^d\times\R^d\to{\mathbb R}$ are some functions depending on $p$ (and $M$) satisfying
\bee\label{eq:main_lemoff1}
|\hat a(p;n)(\bx,\by)| +|\check a(p;n)(\bx,\by)|=O(\rho_n^{p/2+d/2}),\ \ \ \ \ |A_0(n)(\bx,\by)|=O(\rho_n^{d}).
\ene
The constants in the $O$-terms do not depend on $n$ (though they may depend on $M$). They are uniform along every non-degenerate direction when  $|\bx-\by|\asymp 1$. 
\enl

Let us now prove Theorem~\ref{th:main2} assuming Lemma~\ref{main_lemoff} has been proved. First, we obtain expansion \eqref{eq:main_thm2} with the extra constant term and then prove that this constant is, in fact, zero. 

\bel\label{th:main2withconstant}
Suppose the statement of Lemma~\ref{main_lemoff} holds. Then for each $L\in\N$ we have:
\bee\label{eq:main_thm2withconstant}
\begin{split}
&e_{\la}(\bx,\by)=
\cos(\lambda^{1/2}|\bx-\by|)\sum\limits_{q=0}^{2L}\acute a_q(\bx,\by)\lambda^{(d-1)/4-q/2}\cr & +\sin(\lambda^{1/2}|\bx-\by|)\sum\limits_{q=0}^{2L}\grave a_q(\bx,\by)\lambda^{(d-1)/4-q/2}+A_0(\bx,\by)+o(\lambda^{(d-1)/4-L}),
\end{split}
\ene
as $\la\to\infty$. The asymptotic expansion is uniform along every non-degenerate  direction when $|\bx-\by|\asymp 1$. 
\enl

\bep

The proof is similar to the derivation of Theorem \ref{main_thm} from Lemma \ref{main_lem}. However, this time the proof is rather more involved  than the corresponding proof in \cite{PaSh1}, and therefore we write it here in detail. 
Let $M$ be fixed. Denote
\begin{equation}\label{lemoff}
\begin{split}&
\tilde e_n(\rho^2;\bx,\by):=A_0(n)(\bx,\by)+
\cos(\rho|\bx-\by|)\sum\limits_{p=-d+1}^{4M}\hat a(p;n)(\bx,\by)\rho^{-p-(d-1)/2}\\
&+\sin(\rho|\bx-\by|)\sum\limits_{p=-d+1}^{4M}\check a(p;n)(\bx,\by)\rho^{-p-(d-1)/2}.
\end{split}
\end{equation}
Then whenever $\rho\in J_n:=I_{n-1}\cap I_n=[\rho_n,2\rho_n]
$, we have:
\bee
\begin{split}&
\tilde e_n(\rho^2;\bx,\by)-\tilde e_{n-1}(\rho^2;\bx,\by)=\tilde t(n)+\cr &\cos(\rho|\bx-\by|)\sum\limits_{j=-d+1}^{4M}\hat t_j(n)\rho^{-j-(d-1)/2}+\sin(\rho|\bx-\by|)\sum\limits_{j=-d+1}^{4M}\check t_j(n)\rho^{-j-(d-1)/2},
\end{split}
\ene
where \bee\tilde t(n):=A_0(n)-A_0(n-1),\ \ \ 
\hat t_j(n):=\hat a(j;n)-\hat a(j;n-1),\ \ \  \check t_j(n):=\check a(j;n)-\check a(j;n-1).
\ene
On the other hand, since for $\rho\in J_n$ we have \eqref{eq:main_lemoff} for both $n$ and $n-1$, this implies 
\bee
\bes
\tilde t(n)&+\cos(\rho|\bx-\by|)\sum\limits_{j=-d+1}^{4M}\hat t_j(n)\rho^{-j-(d-1)/2}\\
&
+\sin(\rho|\bx-\by|)\sum\limits_{j=-d+1}^{4M}\check t_j(n)\rho^{-j-(d-1)/2}=O(\rho_n^{-M}).
\end{split}
\ene
\becl
For each $j=-d+1,\dots,4M$ we have:
\bee
\tilde t(n)=O(\rho_n^{-M}), \ \ \hat t_j(n)=O(\rho_n^{j+(d-1)/2-M}), \ \  \check t_j(n)=O(\rho_n^{j+(d-1)/2-M}).
\ene
\encl
\bep
Put $$s:=\rho\rho_n^{-1},\ \ \tilde\tau(n):=\tilde t(n)\rho_n^{M},\ \ \hat \tau_j(n):=\hat t_j(n)\rho_n^{M-j-(d-1)/2},\ \ 
\check \tau_j(n):=\check t_j(n)\rho_n^{M-j-(d-1)/2}.$$ Then
\bee\label{Cramer}
\begin{split}&
P(s):=\tilde\tau(n)+\cr &\cos(s\rho_n|\bx-\by|)\sum\limits_{j=-d+1}^{4M}\hat \tau_j(n)s^{-j-(d-1)/2}+\sin(s\rho_n|\bx-\by|)\sum\limits_{j=-d+1}^{4M}\check \tau_j(n)s^{-j-(d-1)/2}=O(1)
\end{split}
\ene
whenever $s\in [1,2]$. 
Now, to show the estimates on coefficients we choose $8M+2d-1$ points in a special way. We put
$$
s_l:=\frac{2\pi}{\rho_n|\bx-\by|}\left(\left[\frac{\rho_n|\bx-\by|}{2\pi}\right]+l\left[\frac{\rho_n|\bx-\by|}{2\pi\cdot 5M}\right]\right),\ \ \ l=1,\dots,4M+d-1,
$$
so that $\sin(s_l\rho_n|\bx-\by|)=0$ and $\cos(s_l\rho_n|\bx-\by|)=1$. We also put 
$$
s'_{l}:=\frac{2\pi}{\rho_n|\bx-\by|}\left(\left[\frac{\rho_n|\bx-\by|}{2\pi}\right]+l\left[\frac{\rho_n|\bx-\by|}{2\pi\cdot  5M}\right]\right)+\frac{\pi}{2\rho_n|\bx-\by|},\ \ \ l=1,\dots,4M+d-1,
$$
so that $\sin(s'_l\rho_n|\bx-\by|)=1$ and $\cos(s'_l\rho_n|\bx-\by|)=0$. Finally, we put
$$
\tilde s:=\frac{2\pi}{\rho_n|\bx-\by|}\left(\left[\frac{\rho_n|\bx-\by|}{2\pi}\right]+(4M+d)\left[\frac{\rho_n|\bx-\by|}{2\pi\cdot 5M}\right]\right)
$$
if $d$ is even (so that there is no $\sin(s\rho_n|\bx-\by|)s^0$ present in \eqref{Cramer}) and
$$
\tilde s:=\frac{2\pi}{\rho_n|\bx-\by|}\left(\left[\frac{\rho_n|\bx-\by|}{2\pi}\right]+(4M+d)\left[\frac{\rho_n|\bx-\by|}{2\pi\cdot 5M}\right]\right)+\frac{\pi}{4\rho_n|\bx-\by|}
$$
if $d$ is odd. We also notice that, assuming $\rho_n$ is sufficiently large, we have $s_{l+1}-s_l\asymp M^{-1}$, $s'_{l+1}-s'_l\asymp M^{-1}$ and $\tilde s-s_{4M+d-1}\asymp M^{-1}$ uniformly in $n$ and $|\bx-\by|\asymp 1$.

{\it Even $d$.} First, we use the points $\{s_l\},\ \tilde s$. Then \eqref{Cramer} and the Cramer's Rule imply that for each $j$ the values $\hat\tau_j(n)$ and $\tilde\tau(n)$ are fractions with a bounded expression in the numerator and a uniform non-zero number in the denominator (the denominator is a Vandermonde determinant). Therefore,
$\hat\tau_j(n)=O(1)$ and $\tilde\tau(n)=O(1)$. Next, we use the points $\{s'_l\}$. Then the estimate for $\tilde\tau(n)$, \eqref{Cramer} and the Cramer's Rule again show that $\check\tau_j(n)=O(1)$. 

{\it Odd $d$.} Again, we use the points $\{s_l\}$ and then the points $\{s'_l\}$. As above, we see that 
\begin{equation}\label{odd1}
\hat\tau_j(n)=O(1),\ \ \ \check\tau_j(n)=O(1)\ \ \  \hbox{for}\  j\not=-(d-1)/2;
\end{equation}
and 
\bee\label{odd2}
\hat\tau_{-(d-1)/2}(n)+\tilde\tau(n)=O(1),\ \ \ \check\tau_{-(d-1)/2}(n)+\tilde\tau(n)=O(1).
\ene
Now, we use the point $\tilde s$ together with \eqref{odd1}. We have 
\bee\label{odd3}
\frac{1}{\sqrt 2}\hat\tau_{-(d-1)/2}(n)+\tilde\tau(n)=O(1),\ \ \ \frac{1}{\sqrt 2}\check\tau_{-(d-1)/2}(n)+\tilde\tau(n)=O(1).
\ene
This and \eqref{odd2} give $\hat\tau_{-(d-1)/2}(n)=O(1)\ $, $\check\tau_{-(d-1)/2}(n)=O(1)\ $, $\tilde\tau(n)=O(1)$.

This shows that $\tilde t(n)=O(\rho_n^{-M})$, $\hat t_j(n)=O(\rho_n^{j+(d-1)/2-M})$ and $\check t_j(n)=O(\rho_n^{j+(d-1)/2-M})$ as claimed.
\enp

Thus, for $j<M-(d-1)/2$, the series $\sum_{m=0}^\infty \hat t_j(m)$ is absolutely convergent; moreover, for such
$j$ we have:
\bee
\begin{split}& 
\hat a(j,n)(\bx,\by)=\hat a(j,0)(\bx,\by)+\sum_{m=1}^n \hat t_j(m)=\hat a(j,0)(\bx,\by)+\sum_{m=1}^\infty \hat t_j(m)+O(\rho_n^{j+(d-1)/2-M})=:\cr & \hat a(j)(\bx,\by)+O(\rho_n^{j+(d-1)/2-M}),
\end{split}
\ene
where we have denoted $\hat a(j)(\bx,\by):=\hat a(j,0)(\bx,\by)+\sum_{m=1}^\infty \hat t_j(m)$. Similarly, for $j<M-(d-1)/2$ we have
\bee
\begin{split}& 
\check a(j,n)(\bx,\by)=\check a(j,0)(\bx,\by)+\sum_{m=1}^n \check t_j(m)=\check a(j,0)(\bx,\by)+\sum_{m=1}^\infty \check t_j(m)+O(\rho_n^{j+(d-1)/2-M})=:\cr & \check a(j)(\bx,\by)+O(\rho_n^{j+(d-1)/2-M}),
\end{split}
\ene
where we have denoted $\check a(j)(\bx,\by):=\check a(j,0)(\bx,\by)+\sum_{m=1}^\infty \check t_j(m)$. Finally, 
\bee
\begin{split}& 
A_0(n)(\bx,\by)=A_0(0)(\bx,\by)+\sum_{m=1}^n \tilde t(m)=A_0(0)(\bx,\by)+\sum_{m=1}^\infty \tilde t(m)+O(\rho_n^{-M})=:\cr & A_0(\bx,\by)+O(\rho_n^{-M}),
\end{split}
\ene
where we have denoted $A_0(\bx,\by):=A_0(0)(\bx,\by)+\sum_{m=1}^\infty \tilde t(m)$.

Since $|\hat a(j,n;\bx,\by)|+|\check a(j,n;\bx,\by)|=O(\rho_n^{j/2+d/2})$
(it was one of the assumptions of lemma), we have:
\bee
\sum_{j=M-(d-1)/2}^{4M}(|\hat a(j,n;\bx,\by)|+|\check a(j,n;\bx,\by)|)\rho_n^{-j-(d-1)/2}=O(\rho_n^{-\frac{M}{3}}),
\ene
assuming as we can without loss of generality that $M$ is sufficiently large (the required `largeness' of $M$ is independent of $\rho_n$). Thus, when $\rho\in I_n$, we have:
\begin{equation}\label{lemoff1}
\begin{split}&
e_{\rho^2}(\bx,\by)=
\cos(\rho|\bx-\by|)\sum\limits_{p=-d+1}^{M-(d-1)/2-1}\hat a(p)(\bx,\by)\rho^{-p-(d-1)/2}+\cr & \sin(\rho|\bx-\by|)\sum\limits_{p=-d+1}^{M-(d-1)/2-1}\check a(p)(\bx,\by)\rho^{-p-(d-1)/2}+A_0(\bx,\by)+O(\rho_n^{-M/3}).
\end{split}
\end{equation}
Since constants in $O$ do not depend on $n$, 
for all $\rho\ge \rho_0$ we have:
\begin{equation}\label{lemoff2}
\begin{split}&
e_{\rho^2}(\bx,\by)=
\cos(\rho|\bx-\by|)\sum\limits_{p=-d+1}^{[M/3]-d+1}\hat a(p)(\bx,\by)\rho^{-p-(d-1)/2}+\cr & \sin(\rho|\bx-\by|)\sum\limits_{p=-d+1}^{[M/3]-d+1}\check a(p)(\bx,\by)\rho^{-p-(d-1)/2}+A_0(\bx,\by)+O(\rho_n^{-\frac{M}{3}+\frac{d-1}{2}}).
\end{split}
\end{equation}
Taking $M=6L+1$ and making change $q=p+d-1$, we obtain \eqref{eq:main_thm2withconstant}. 
\enp

\bel
For all $\bx\ne\by$ we have $A_0(\bx,\by)=0$.
\enl
\bep
The proof is similar to the approach Shenk and Shubin \cite{ShuSche} used to get rid of the constant in the one-dimensional case (strangely enough, they used this trick on the diagonal; the trick they used for similar purpose off the diagonal does not work in high dimensions). 

First, we notice that, without loss of generality, we can assume that 
the spectrum of $H$ is contained in $[2,+\infty)$. Indeed, if this is not the case, we consider instead the operator $H+sI$ with sufficiently large $s$; it is easy to see that this change does not affect the constant in \eqref{eq:main_thm2withconstant}. 

Let us construct the complex powers of $H$. For $\Re z<-d/2$, the operator $H^z$ has the integral kernel $K_z(\bx,\by)$ holomorphic in $z$; the Seeley type theorem (see \cite{Shu1,Shu2}) then implies that this kernel can be meromorphically continued to the entire complex plane; moreover, 
$K_0(\bx,\by)=0$ when $\bx\ne\by$. For $\Re z<-d/2$, we have:
\bee\label{eq:kernel}
K_z(\bx,\by)=\int_1^{\infty}\la^zd_{\la}e(\la;\bx,\by)=
\int_1^{\infty}z\la^{z-1}e(\la;\bx,\by)d\la. 
\ene 
If we plug \eqref{eq:main_thm2withconstant} with $L=d$ into the RHS of \eqref{eq:kernel}, we will see that the value at $z=0$ of the meromorphic continuation of all the terms in the RHS, except $A_0$, will be zero, so we have $A_0(\bx,\by)=K_0(\bx,\by)=0$.  
\enp

\ber
As we have mentioned in the introduction, we also could have used the heat asymptotic expansion of \cite{Mol} to get rid of the constant $A_0$. 
\enr


The rest of the paper is devoted to proving Lemmas \ref{main_lem} and \ref{main_lemoff}. 

\section{Perturbation of the spectral function}
In this section, we study the spectral projections of two self-adjoint operators $H_1$ and $H_2$ 
that are sufficiently close to each other and compare them. 
We assume that both $H_1$ and $H_2$ act in a Hilbert space $\GH$ and are bounded below: $H_j>aI$. These operators will be assumed to be close not just in the usual operator norm, but also in the abstract version of the Sobolev norm. More precisely, we fix a number $s\ge 0$ and assume that 
\bee\label{4.1}
||(H_1-H_2)(H_2+(1-a)I)^s||<\ep<1.
\ene

For any self-adjoint operator $H$ and any Borel set $I\subset \R$ we denote by 
\bee
E(I;H)
\ene
the spectral projection of $H$ corresponding to the set $I$. We also put
\bee
E_{\la}(H):=E((-\infty,\la];H).
\ene
Let $f\in\GH$. We want to prove that $E_{\la}(H_2)f-E_{\la}(H_1)f$ is small.
Let $\de\ge\ep$ (later we will put $\de=\ep^{1/2}$). 
\bel\label{lem:1} 
Suppose that $\la-a\ge 1$. Then 
we have
\bee
||E((-\infty,\la-\de];H_1)E([\la+\de,+\infty);H_2)(H_2-a+1)^s||\le\frac{\pi\ep}{\de}.
\ene
\enl
\bep

Let us assume that 
\bee
\phi=E((-\infty,\la-\de];H_1)\phi, \ \ \ 
(H_2-a+1)^s\psi=E((\la+\de,\infty];H_2)(H_2-a+1)^s\psi,  
\ene
and 
$||\phi||=||\psi||=1$. The statement is equivalent to proving 
$|(\phi,(H_2-a+1)^s\psi)|\le \pi\ep/\de$. 
Denote by $\gamma=\gamma_N$ the closed square contour in the complex plane symmetric about the real axis and intersecting it at two points: $\lambda$ and $-N$, where $N>-a$ is a large number. Then we have 
\bee\label{contour}
\phi=\int_{\gamma}(H_1-z)^{-1}\phi\, dz.
\ene
Note that the integral $\int_{\gamma}(H_1-z)^{-1}dz$ does not need to converge in general, so we understand the integral in the RHS of \eqref{contour} in the strong sense. Now we have:
\bee\label{contour1}
\bes
&
(\phi,(H_2-a+1)^s\psi)=\int_{\gamma}((H_1-z)^{-1}\phi,(H_2-a+1)^s\psi)dz\\
&
=
\int_{\gamma}(\phi,(H_1-\bar z)^{-1}(H_2-a+1)^s\psi)dz\\
&
=\int_{\gamma}(\phi,((H_2-\bar z)^{-1}+(H_1-\bar z)^{-1}(H_1-H_2)(H_2-\bar z)^{-1})(H_2-a+1)^s\psi)dz\\
&
=
\int_{\gamma}(\phi,(H_1-\bar z)^{-1}(H_1-H_2)(H_2-a+1)^s(H_2-\bar z)^{-1}\psi)dz;
\end{split}
\ene
in the last line we have used the fact that (here, $\bar\gamma$ is the contour complex conjugated to $\gamma$; in fact $\bar\gamma=\gamma$)
\bee
\int_{\bar\gamma}(H_2-\bar z)^{-1}E((\la+\de,\infty];H_2)d\bar z=0.
\ene

Therefore, we have:
\bee\label{contour2}
\bes
&
|(\phi,\psi)|\le\ep
(\int_{\gamma}||(H_1-z)^{-1}\phi||^2 \ |dz|)^{1/2}( \int_{\gamma}||(H_2-\bar z)^{-1}\psi||^2 \ |dz|)^{1/2}\\
&\le\ep
(\int_{\gamma}||(H_1-z)^{-1}||^2 \ |dz|)^{1/2}( \int_{\gamma}||(H_2-\bar z)^{-1}||^2 \ |dz|)^{1/2}. 
\end{split}
\ene
Now the estimate follows from the spectral theorem, since it implies that
\bee
\lim_{N\to\infty}\int_{\gamma_N}||(H_1-z)^{-1}||^2 \ |dz|\le\frac{\pi}{\de}
\ene
and
\bee
\lim_{N\to\infty}\int_{\gamma_N}||(H_2-\bar z)^{-1}||^2 \ |dz|\le\frac{\pi}{\de}.
\ene

\enp

Notice that for $s=0$ Lemma~\ref{lem:1} is a simple version of the Davis-Kahan $\sin\Theta$ Theorem (see e.g. Thms VII.3.1-3.4 from \cite{Bh}). 

\bel \label{lem:2}
Let $f\in\GH$. 
Under the above assumptions, we have:
\bee\label{eq:projdifference}
\bes
&||E_{\la}(H_2)f-E_{\la}(H_1)f||\le 2||E([\la-\de,\la+\de];H_2)f||\\
&+2\pi\ep\de^{-1}||E((-\infty,\lambda];H_2)f||+2\pi\ep\de^{-1}||(H_2-a+1)^{-s}f||.
\end{split}
\ene
\enl

\bep
We have:
\bee\label{trick1}
\bes
&
E_{\la}(H_2)f=E((-\infty,\la-\de);H_2)f+E([\la-\de,\la];H_2)f\\
&
=\bigl[E((-\infty,\la];H_1)+E((\la,+\infty);H_1)\bigr]E((-\infty,\la-\de];H_2)f+E([\la-\de,\la];H_2)f.
\end{split}
\ene
To estimate the second term in the RHS we use Lemma~\ref{lem:1} with $s=0$: 
\bee
\bes
&||E((\la,+\infty);H_1)E((-\infty,\la-\de];H_2)f||\\
&=
||E((\la,+\infty);H_1)E((-\infty,\la-\de];H_2)E((-\infty,\la-\de];H_2)f||\\
&\le 2\pi\ep\de^{-1}||E((-\infty,\lambda];H_2)f||.
\end{split}
\ene

For the first term we perform the same trick as in \eqref{trick1} again, and obtain:
\bee\label{trick2}
\bes
&
E((-\infty,\la];H_1)E((-\infty,\la-\de];H_2)f\\
&
=E((-\infty,\la];H_1)f-E((-\infty,\la];H_1)E((\la-\de,\la+\de);H_2)f\\
&
-
E((-\infty,\la];H_1)E([\la+\de,+\infty);H_2)(H_2-a+1)^s(H_2-a+1)^{-s}f. 
\end{split}
\ene
Now the statement follows from \eqref{trick1}, \eqref{trick2}, and Lemma \ref{lem:1}. 
\enp

Suppose, 
\bee
H=-\Delta+V
\ene
acting in $\R^d$, where $V\in \mathrm{USB}(\R^d)$. Then $E_{\la}(H)$ has the Schwartz kernel (see e.g. \cite{AgKa,Si}) which we denote by $e(\la;H;\bx,\by)$; we will often omit writing the dependence on some of the arguments when it could cause no ambiguity. Let $\bx_0$ be any point from $\R^d$; denote by $\de_{\bx_0}$ the Dirac delta-function centred at $\bx_0$. 
\bel\label{L2norm}
$E_{\la}(H)\de_{\bx_0}\in L_2(\R^d)$ and for large $\la$ we have 
\bee
||E_{\la}(H)\de_{\bx_0}|| \ll \la^{d/4}. 
\ene
\enl
\bep
We have:
\bee
\bes
||E_{\la}(H)\de_{\bx_0}||^2=&\int_{\R^d}|e(\la;H;\bx_0,\by)|^2 d\by
=
\int_{\R^d}e(\la;H;\bx_0,\by) e(\la;H;\by,\bx_0) d\by\\
&
=e(\la;H;\bx_0,\bx_0)=O(\la^{d/2}). 
\end{split}
\ene
We used the fact that $e$ is real-valued and symmetric and $E$ is a projection. For the last estimate see e.g. \cite {AgKa}. \enp

Finally, we will prove another Lemma which we will need in Section 7. 

\begin{defn}\label{Tau} For an interval $[\lambda_1,\lambda_2]\subset\R$ we introduce the set $T([\lambda_1,\lambda_2])$ of functions $\tau(z)$ satisfying the following properties:

1) $\tau(z)$ is analytic in the neighbourhood of the interval $[\lambda_1,\lambda_2]$,

2) $\tau(x)$ is real-valued on the interval $[\lambda_1,\lambda_2]$,

3) $\tau'(x)>0$ on the interval $[\lambda_1,\lambda_2]$.
\end{defn}

In particular, these properties imply that the function $\tau^{-1}$ is well-defined and analytic in a neighborhood of the interval $[\tau(\lambda_1),\tau(\lambda_2)]$. We will need the following auxiliary statement which can be considered as a special case of the rule for changing the order of the integration.
 
\begin{lem}\label{aux} Let $\tau\in T$. Let function $f$ be analytic in a neighborhood of $[\tau(\lambda_1),\tau(\lambda_2)]$. Let $\Gamma$ be a contour around interval $[\lambda_1,\lambda_2]$ completely inside the domain of analyticity of the function $\tau$ and let $\tilde\Gamma$ be a contour around interval $[\tau(\lambda_1),\tau(\lambda_2)]$ completely inside the domain of analyticity of the functions $\tau^{-1}$ and $f$. Then we have the following identity
\begin{equation}\label{intaux}
\int\limits_{\tau(\lambda_1)}^{\tau(\lambda_2)}f(r)\,dr\oint\limits_{\Gamma}\frac{d\mu}{\tau^{-1}(r)-\mu}=-\int\limits_{\lambda_1}^{\lambda_2}\,d\mu\oint\limits_{\tilde\Gamma}\frac{f(r)\,dr}{\tau^{-1}(r)-\mu}.
\end{equation}
\end{lem}
\begin{proof} Denote by $I$ the left hand side of \eqref{intaux}. If necessary, shrinking by analyticity $\Gamma$ and $\tilde\Gamma$, we may assume that $\tilde\Gamma=\tau(\Gamma)$. We introduce the change of  variables: $\nu:=\tau^{-1}(r)$, $s:=\tau(\mu)$. Then we have
$$
I=\int\limits_{\tau(\lambda_1)}^{\tau(\lambda_2)}dr\oint\limits_{\Gamma}\frac{f(\tau(\mu))\,d\mu}{\tau^{-1}(r)-\mu}=-\int\limits_{\lambda_1}^{\lambda_2}\tau'(\nu)\,d\nu\oint\limits_{\tilde\Gamma}\frac{f(s)(\tau^{-1})'(s)\,ds}{\tau^{-1}(s)-\nu}.
$$
Next, by the properties of $\tau$ and Cauchy Theorem we have
\begin{equation*}
\begin{split}
&-\int\limits_{\lambda_1}^{\lambda_2}\,d\nu\oint\limits_{\tilde\Gamma}\frac{f(s)\tau'(\nu)(\tau^{-1})'(s)\,ds}{\tau^{-1}(s)-\nu}=-\int\limits_{\lambda_1}^{\lambda_2}\,d\nu\oint\limits_{\tilde\Gamma}\frac{f(s)\tau'(\tau^{-1}(s))(\tau^{-1})'(s)\,ds}{\tau^{-1}(s)-\nu}\cr =&-\int\limits_{\lambda_1}^{\lambda_2}\,d\nu\oint\limits_{\tilde\Gamma}\frac{f(s)\,ds}{\tau^{-1}(s)-\nu}.
\end{split}
\end{equation*}
Now, the change of notation $\nu$ by $\mu$ and $s$ by $r$ completes the proof. \end{proof}

\section{Resonance zones and coordinates there}

In this section, we define resonance regions, state some of their properties and introduce convenient coordinates in these zones. The material in this section follows the narration of \cite{PaSh1} which contains the proofs of all statements in this section. 

Recall the definition of the quasi-lattice subspaces from Section 2: we say that $\GV$ is a quasi-lattice subspace of dimension $m$, if $\GV$ is a linear
span of $m$ linear independent vectors $\bth_1,\dots,\bth_m$ with $\bth_j\in\tilde\Bth\ \forall j$. As before, by
$\Bth_{\tilde k}$
we denote the algebraic sum of $\tilde k$ copies of $\Bth$; remember that we consider the index $\tilde k$ fixed. We also put
$\Bth'_{\tilde k}:=\Bth_{\tilde k}\setminus\{0\}$. For each $\GV\in\CV$ we put $S_{\GV}:=\{\bxi\in\GV,\ |\bxi|=1\}$.
For each non-zero $\bth\in\R^d$ we put $\bn(\bth):=\bth|\bth|^{-1}$.

Let $\GV\in\CV_m$. We say that $\GF$ is a {\it flag} generated by $\GV$, if $\GF$ is a sequence $\GV_j\in\CV_j$ ($j=0,1,\dots,m$)
such that $\GV_{j-1}\subset\GV_j$ and $\GV_m=\GV$. We say that $\{\bnu_j\}_{j=1}^m$ is a sequence generated by $\GF$
if $\bnu_j\in\GV_j\ominus\GV_{j-1}$ and $||\bnu_j||=1$ (obviously, this condition determines each $\bnu_j$ up to the
multiplication by $-1$). We denote by $\CF(\GV)$ the collection of all flags generated by $\GV$. We also
fix an increasing sequence of positive numbers $\al_j$ ($j=1,\dots,d$) with $\al_d<\frac{1}{2 d}$
(these numbers depend only on $d$) and put $L_j:=\rho_n^{\al_j}$.

Let $\bth\in\Bth'_{\tilde k}$. We call by {\it resonance zone} generated by $\bth$
\bee\label{eq:1}
\Lambda(\bth):=\{\bxi\in\R^d,\ |\lu\bxi,\bn(\bth)\ru|\le L_1\}.
\ene
Suppose, $\GF\in\CF(\GV)$ is a flag and $\{\bnu_j\}_{j=1}^m$ is a sequence generated by $\GF$. We define
\bee\label{eq:2}
\Lambda(\GF):=\{\bxi\in\R^d,\ |\lu\bxi,\bnu_j\ru|\le L_j\}.
\ene
If $\dim\GV=1$, definition \eqref{eq:2} is reduced to \eqref{eq:1}.
Obviously, if $\GF_1\subset\GF_2$, then $\Lambda(\GF_2)\subset\Lambda(\GF_1)$.

Suppose, $\GV\in\CV_j$. We denote
\bee\label{eq:Bxi1}
\Bxi_1(\GV):=\cup_{\GF\in\CF(\GV)}\Lambda(\GF).
\ene
Note that
$\Bxi_1(\GX)=\R^d$ and $\Bxi_1(\GV)=\Lambda(\bth)$ if $\GV\in\CV_1$ is spanned by $\bth$.
Finally, we put
\bee\label{eq:Bxi}
\Bxi(\GV):=\Bxi_1(\GV)\setminus(\cup_{\GU\supsetneq\GV}\Bxi_1(\GU))=\Bxi_1(\GV)
\setminus(\cup_{\GU\supsetneq\GV}\cup_{\GF\in\CF(\GU)}\Lambda(\GF)).
\ene
We call $\Bxi(\GV)$ the resonance region generated by $\GV$. 
Sometimes, we will be calling the region $\Bxi(\GX)$ {\it the non-resonance region}.

The following results were proved in \cite{PaSh1}. We always assume that $\rho_0$
(and thus $\rho_n$) is sufficiently large.

\bel\label{lem:propBUps}

(i) We have
\bee
\cup_{\GV\in\CV}\Bxi(\GV)=\R^d.
\ene

(ii) $\bxi\in\Bxi_1(\GV)$ iff $\bxi_{\GV}\in\Omega(\GV)$, where $\Omega(\GV)\subset\GV$
is a certain bounded set (more precisely, $\Omega(\GV)=\Bxi_1(\GV)\cap\GV\subset B(m L_m)$ if
$\dim\GV=m$).

(iii) $\Bxi_1(\R^d)=\Bxi(\R^d)$ is a bounded set, $\Bxi(\R^d)\subset B(d L_d)$; all other sets
$\Bxi_1(\GV)$ are unbounded.

\enl

\bel\label{lem:intersect}
Let $\GV,\GU\in\CV$. Then $(\Bxi_1(\GV)\cap\Bxi_1(\GU))\subset \Bxi_1(\GW)$,
where $\GW:=\GV+\GU$ (algebraic sum).
\enl
\bec
(i) We can re-write definition \eqref{eq:Bxi} like this:
\bee\label{eq:Bxibis}
\Bxi(\GV):=\Bxi_1(\GV)\setminus(\cup_{\GU\not\subset\GV}\Bxi_1(\GU)).
\ene

(ii) If $\GV\ne\GU$, then $\Bxi(\GV)\cap\Bxi(\GU)=\emptyset$.

(iii) We have $\R^d=\sqcup_{\GV\in\CV}\Bxi(\GV)$ (the disjoint union).
\enc

\bel\label{lem:newXi}
We have
\bee\label{eq:newBxi}
\Bxi_1(\GV)\cap\cup_{\GU\supsetneq\GV}\Bxi_1(\GU)=\Bxi_1(\GV)\cap\cup_{\GW\supsetneq\GV, \dim\GW=1+\dim\GV}\Bxi_1(\GW).
\ene
\enl

\bec
We can re-write \eqref{eq:Bxi} as
\bee\label{eq:Bxibbis}
\Bxi(\GV):=\Bxi_1(\GV)\setminus(\cup_{\GW\supsetneq\GV, \dim\GW=1+\dim\GV}\Bxi_1(\GW)).
\ene
\enc

\bel\label{lem:Upsilon}
Let $\GV\in\CV$ and $\bth\in\Bth_{\tilde k}$. Suppose that $\bxi\in\Bxi(\GV)$ and
both points $\bxi$ and $\bxi+\bth$ are inside $\Lambda(\bth)$. Then $\bth\in\GV$ and $\bxi+\bth\in\Bxi(\GV)$.
\enl

Now we define another important object. 

\begin{defn}
\label{reachability:defn}
Let $\bth, \bth_1, \bth_2, \dots, \bth_l$ be some vectors from $\Bth'_{\tilde k}$,
which are not necessarily distinct.
\begin{enumerate}
\item \label{1}
We say that two vectors
$\bxi, \boldeta\in\R^d$ are \textsl{$\bth$-resonant congruent}
if both $\bxi$ and $\boldeta$ are inside $\L(\bth)$ and $(\bxi - \boldeta) =l\bth$ with $l\in\Z$.
In this case we write $\bxi \leftrightarrow \boldeta \mod \bth$.
\item\label{2}
For each $\bxi\in\R^d$ we denote by
$\BUps_{\bth}(\bxi)$ the set of all points which are
$\bth$-resonant congruent to $\bxi$.
For $\bth\not = \mathbf 0$ we say that
$\BUps_{\bth}(\bxi) = \varnothing$
if $\bxi\notin\L(\bth)$.
\item\label{3}
We say that $\bxi$ and $\boldeta$
are \textsl{$\bth_1, \bth_2, \dots, \bth_l$-resonant congruent},
if there exists a sequence $\bxi_j\in\R^d, j=0, 1, \dots, l$ such that
$\bxi_0 = \bxi$, $\bxi_l = \boldeta$,
and $\bxi_j\in\BUps_{\bth_j}(\bxi_{j-1})$ for $j = 1, 2, \dots, l$.

\item
We say that $\boldeta\in\R^d$ and
$\bxi\in\R^d$ are \textsl{resonant congruent}, if either $\bxi=\boldeta$ or $\bxi$ and $\boldeta$ are $\bth_1, \bth_2, \dots, \bth_l$-resonant congruent with some
 $\bth_1, \bth_2, \dots, \bth_l\in\Bth_{\tilde k}'$. The set of \textbf{all} points, resonant congruent
to $\bxi$, is denoted by $\BUps(\bxi)$.
For points $\boldeta\in\BUps(\bxi)$ (note that this condition is equivalent to
$\bxi\in\BUps(\boldeta)$) we write $\boldeta\leftrightarrow\bxi$.
\end{enumerate}
\end{defn}

Note that $\BUps(\bxi)=\{\bxi\}$ for any $\bxi\in\Bxi(\GX)$. Now Lemma \ref{lem:Upsilon} immediately implies
\bec\label{cor:Upsilon}
For each $\bxi\in\Bxi(\GV)$ we have $\BUps(\bxi)\subset\Bxi(\GV)$ and thus
\bees
\Bxi(\GV)=\sqcup_{\bxi\in\Bxi(\GV)}\BUps(\bxi).
\enes
\enc
\bel\label{lem:diameter}
The diameter of $\BUps(\bxi)$ is bounded above by $mL_m$, if $\bxi\in\Bxi(\GV)$, $\GV\in\CV_m$.
\enl
\bel\label{lem:finiteBUps}
For each $\bxi\in\Bxi(\GV),\ \GV\ne\R^d$, the set $\BUps(\bxi)$ is finite.
\enl

Next, we are going to introduce a convenient set of coordinates in $\Bxi(\GV)$.
Let $\GV\in\CV_m$, $m<d$, be fixed. Then, as we know, $\bxi\in\Bxi_1(\GV)$ if and only if
$\bxi_{\GV}\in\Om(\GV)$. Let $\{\GU_j\}$ be a collection of all subspaces $\GU_j\in\CV_{m+1}$
such that each $\GU_j$ contains $\GV$. Let $\bmu_j=\bmu_j(\GV)$ be (any) unit vector from $\GU_j\ominus\GV$.
Then, for $\bxi\in\Bxi_1(\GV)$, we have  $\bxi\in\Bxi_1(\GU_j)$ if and only if the estimate
$|\lu\bxi,\bmu_j\ru|=|\lu\bxi_{\GV^{\perp}},\bmu_j\ru|\le L_{m+1}$ holds. Thus, formula \eqref{eq:Bxibbis} implies that
\bee\label{eq:Bxibbbis}
\Bxi(\GV)=\{\bxi\in\R^d,\ \bxi_{\GV}\in\Omega(\GV)\ \&\ \forall j \ |\lu\bxi_{\GV^{\perp}},\bmu_j(\GV)\ru| > L_{m+1}\}.
\ene
The collection $\{\bmu_j(\GV)\}$ obviously coincides with
\bee
\{\bn(\bth_{\GV^{\perp}}),\ \bth\in\Bth_{\tilde k}\setminus\GV\}.
\ene

The set $\Bxi(\GV)$ is, in general, disconnected; it consists of several
connected components which we will denote by $\{\Bxi(\GV)_p\}_{p=1}^P$. Let us fix a connected component $\Bxi(\GV)_p$.
Then for some vectors $\{\tilde\bmu_j(p)\}_{j=1}^{J_p}\subset \{\pm\bmu_j\}$ we have
\bee\label{eq:Bxip}
\Bxi(\GV)_p=\{\bxi\in\R^d,\ \bxi_{\GV}\in\Omega(\GV)\ \&\ \forall j \ \lu\bxi_{\GV^{\perp}},\tilde\bmu_j(p)\ru > L_{m+1}\};
\ene
we assume that $\{\tilde\bmu_j(p)\}_{j=1}^{J_p}$ is the minimal set with this property, so that each hyperplane
$$
\{\bxi\in\R^d,\ \bxi_{\GV}\in\Omega(\GV)\ \ \&\ \ \lu\bxi_{\GV^{\perp}},\tilde\bmu_j(p)\ru = L_{m+1}\},\ j=1,\dots,J_p
$$
has a non-empty intersection with the boundary of $\Bxi(\GV)_p$. It is not hard to see that $J_p\ge d-m$. Indeed, otherwise $\Bxi(\GV)_p$ would have non-empty intersection with $\Bxi_1(\GV')$ for some $\GV'$, $\GV\subsetneq\GV'$. We also introduce
\bee\label{eq:Bxiptilde}
\tilde\Bxi(\GV)_p:=\{\bxi\in\GV^{\perp},\ 
\forall j \
\lu\bxi,\tilde\bmu_j(p)\ru > 0\}.
\ene
Note that our assumption that $\Bxi(\GV)_p$ is a connected component of $\Bxi(\GV)$ implies that for any
$\bxi\in\tilde\Bxi(\GV)_p$ and any $\bth\in\Bth_{\tilde k}\setminus\GV$ we have
\bee\label{eq:ne0}
\lu\bxi,\bth\ru=\lu\bxi,\bth_{\GV^{\perp}}\ru\ne 0.
\ene
We also put $K:=d-m-1$.

Throughout most of this paper, we will assume that the number $J_p$ of `defining planes' is the minimal possible, i.e. $J_p= K+1$; in the general case we refer to Section 11 of \cite{PaSh1} where it is explained how to deal with the case of arbitrary
$\Xi(\GV)_p$.
If $J_p= K+1$, then the set
$\{\tilde\bmu_j(p)\}_{j=1}^{K+1}$ is linearly independent.
Let $\ba=\ba(p)$ be a unique
point from $\GV^\perp$ satisfying the following conditions: $\lu\ba,\tilde\bmu_j(p)\ru = L_{m+1}$, $j=1,\dots,K+1$.
Then, since the determinant of the Gram matrix of vectors $\tilde\bmu_j(p)$ is $\gg\rho_n^{0-}$, we have $|\ba|\ll L_{m+1}\rho_n^{0+}$.
We introduce the shifted cylindrical coordinates in $\Bxi(\GV)_p$. These coordinates will be denoted by $\bxi=(r;\tilde\Phi;X)$. Here,
$X=(X_1,\dots,X_m)$ is an arbitrary set of cartesian coordinates in $\Omega(\GV)$. These coordinates do not depend on the choice of the
connected component $\Bxi(\GV)_p$. The rest of the coordinates $(r,\tilde\Phi)$ are shifted spherical coordinates in $\GV^{\perp}$,
centered at $\ba$. This means that
\bee\label{eq:r}
r(\bxi)=|\bxi_{\GV^{\perp}}-\ba|
\ene
and
\bee\label{eq:tildePhi}
\tilde\Phi=\bn(\bxi_{\GV^{\perp}}-\ba)\in S_{\GV^{\perp}}.
\ene
More precisely, $\tilde\Phi\in M$, where $M=M_p:=\{\bn(\bxi_{\GV^{\perp}}-\ba),\ \bxi\in\Bxi(\GV)_p\}\subset S_{\GV^{\perp}}$
is a $K$-dimensional spherical simplex with $K+1$ sides. Note that
\bee\label{eq:Mp}
\bes
M_p&=\{\bn(\bxi_{\GV^{\perp}}-\ba),\ \bxi\in\Bxi(\GV)_p\}
=\{\bn(\bxi_{\GV^{\perp}}-\ba),\  \forall j \ \lu\bxi_{\GV^{\perp}},\tilde\bmu_j(p)\ru > L_{m+1}\}\\
&=\{\bn(\boldeta),\ \boldeta:=\bxi_{\GV^{\perp}}-\ba\in\GV^\perp,\  \forall j \ \lu\boldeta,\tilde\bmu_j(p)\ru > 0\}
=S_{\GV^{\perp}}\cap\tilde\Bxi(\GV)_p.
\end{split}
\ene
We will 
denote by $d\tilde\Phi$ the spherical Lebesgue measure on $M_p$.
For each non-zero vector $\bmu\in\GV^{\perp}$, we denote
\bee
W(\bmu):=\{\boldeta\in\GV^\perp,\ \lu\boldeta,\bmu\ru=0\}.
\ene
Thus, the sides of the simplex $M_p$ are intersections of $W(\tilde\bmu_j(p))$ with the sphere $S_{\GV^{\perp}}$.
Each vertex $\bv=\bv_t$, $t=1,\dots,K+1$ of $M_p$ is an intersection of $S_{\GV^{\perp}}$ with $K$ hyperplanes
$W(\tilde\bmu_j(p))$,
$j=1,\dots,K+1$, $j\ne t$. This means that $\bv_t$ is a unit vector from $\GV^{\perp}$ which is orthogonal to $\{\tilde\bmu_j(p)\}$,
$j=1,\dots,K+1$, $j\ne t$; this defines $\bv$ up to a multiplication by $-1$.
\bel\label{lem:sign}
Let $p$ be fixed.
Suppose, $\bth\in\Bth_{\tilde k}\setminus\GV$ and $\bth_{\GV^{\perp}}=\sum_{j=1}^{K+1} b_j\tilde\bmu_j(p)$.
Then either all coefficients $b_j$ are non-positive, or all of them are non-negative.
\enl

Assume that the diameter of $M_p$ is $\le (100d^2)^{-1}$, which we can always achieve by taking sufficiently large $\tilde k$. We put $\Phi_q:=\frac{\pi}{2}-\phi(\bxi_{\GV^\perp}-\ba,\tilde\bmu_q(p))$, $q=1,\dots,K+1$, where $\phi(\cdot,\cdot)$ is the angle between two non-zero vectors. 
The geometrical meaning of these coordinates is simple: $\Phi_q$ is the spherical distance between
$\tilde\Phi=\bn(\bxi_{\GV^{\perp}}-\ba)$ and $W(\tilde\bmu_q(p))$. The reason why we have introduced $\Phi_q$ is that in these coordinates some important objects will be especially simple (see e.g. Lemma~\ref{lem:products} below) which is very convenient for integration in Section 7. At the same time, the set of coordinates $(r,\{\Phi_q\})$ contains $K+2$ variables, whereas we only need $K+1$ coordinates in
$\GV^{\perp}$. Thus, we have one constraint for variables $\Phi_j$. Namely, 
let $\{\be_j\}$, $j=1,\dots,K+1$ be a fixed orthonormal basis in $\GV^{\perp}$ chosen in such a way that the $K+1$-st axis
passes through $M_p$.
Then we have $\be_j=\sum_{l=1}^{K+1}a_{jl}\tilde\bmu_l$ with some matrix $\{a_{jl}\}$, $j,l=1,\dots,K+1$, and $\tilde\bmu_l=\tilde\bmu_l(p)$.
Therefore (recall that we denote $\boldeta:=\bxi_{\GV^{\perp}}-\ba$),
\bee\label{eq:etaj}
\eta_j=\lu\boldeta,\be_j\ru=r\sum_{q=1}^{K+1}a_{jq}\sin\Phi_q
\ene
and, since $r^2(\bxi)=|\boldeta|^2=\sum_{j=1}^{K+1}\eta_j^2$, this
implies that
\bee\label{eq:odin}
\sum_j(\sum_q a_{jq}\sin\Phi_q)^2=1,
\ene
which is our constraint. 

Let us also put
\bee\label{eq:etajdash}
\eta_j':=\frac{\eta_j}{|\boldeta|}=\sum_{q=1}^{K+1}a_{jq}\sin\Phi_q.
\ene
Then we can write the surface element $d\tilde\Phi$ in the coordinates $\{\eta_j'\}$ as
\bee\label{eq:surfaceelement}
d\tilde\Phi=\frac{d\eta_1'\dots d\eta_K'}{\eta_{K+1}}=\frac{d\eta_1'\dots d\eta_K'}{(1-\sum_{j=1}^K(\eta_{j}')^2)^{1/2}},
\ene
where the denominator is bounded below by $1/2$ by our choice of the basis $\{\be_j\}$.

The next lemma describes the dependence on $r$ of all possible inner products $\lu\bxi,\bth\ru$, $\bth\in\Bth_{\tilde k}$, $\bxi\in\Bxi(\GV)_p$.
\bel\label{lem:products}
Let
$\bxi\in\Bxi(\GV)_p$, $\GV\in\CV_m$, and $\bth\in\Bth_{\tilde k}$.

(i) If $\bth\in\GV$, then $\lu\bxi,\bth\ru$ does not depend on $r$.

(ii) If $\bth\not\in\GV$ and $\bth_{\GV^{\perp}}=\sum_{q}b_q\tilde\bmu_q(p)$, 
then
\bee\label{eq:innerproduct1}
\lu\bxi,\bth\ru=\lu X,\bth_{\GV}\ru+L_{m+1}\sum_{q}b_q+r(\bxi)\sum_{q}b_q\sin\Phi_q.
\ene

In the case (ii) all the coefficients $b_q$ are either non-positive or non-negative and
each non-zero coefficient $b_q$ satisfies
\bee\label{eq:n10}
 |b_q| \asymp 1.
\ene
\enl

Finally, let us denote (recall that $\la_n=\rho_n^2$)
\bee
\CX_n
:=\{\bxi\in\R^2,\,|\bxi|^2\in [0.7\la_n,
17.5\la_n
]\}.
\ene
We also put
\bee
\CA=\CA_n:=\cup_{\bxi\in\CX_n}\BUps(\bxi).
\ene
Lemma \ref{lem:diameter} implies that for each $\bxi\in\CA$ we have $|\bxi|^2\in[0.5\la_n,18\la_n]$. In particular, we have
\bee\label{eq:CARd}
\CA\cap\BXi(\R^d)=\emptyset.
\ene
For each $\GV\in\CV_m$, $m<d$, we put
\bee
\CA(\GV):=\CA_n\cap\Bxi(\GV).
\ene

%

\section{Pseudo-differential operators and the gauge transform}

This section is another one where we present definitions and results from \cite{PaSh1}; as before, the proofs of all statements can be found either in that paper or in \cite{Sob}, \cite{Sob1}, and \cite{ParSob}.
In this section, we construct operators $H_1$ and $H_2$ described in the Introduction. Since we have agreed that our potential is quasi-periodic, it is enough for our purposes to deal with quasi-periodic pseudo-differential operators. 

\subsection{Classes of PDO's and their properties}\label{classes:subsect}


For any $f\in L_2(\R^d)$ we define the Fourier transform:
\begin{equation*}
(\mathcal F f)(\bxi)
= \frac{1}{(2\pi)^{\frac{d}{2}}} \int_{\R^d} e^{-i\bxi
\bx}f(\bx) d\bx,\ \bxi\in\R^d.
\end{equation*}
Let $b = b(\bx, \bxi)$, $\bx, \bxi\in\R^d$, be a
quasi-periodic (in $\bx$) complex-valued function, i.e. for some finite  set $\hat{\Bth}$ of frequencies (we always assume $\hat\Bth$ to be symmetric and to contain $0$)
\begin{equation}\label{eq:sumf}
b(\bx, \bxi) = \sum\limits_{\bth\in\hat{\Bth}}\hat{b}(\bth, \bxi)\be_{\bth}(\bx)
\end{equation}
where
\bees
\hat{b}(\bth, \bxi):=\BM_\bx(b(\bx,\bxi)\be_{-\bth}(\bx))
\enes
are Fourier coefficients of $b$ (recall that $\BM$ is the mean of an almost-periodic function). 
Put $\lu \bt \ru := \sqrt{1+|\bt|^2},\
\forall \bt\in\R^d$. 
We say that the symbol $b$ belongs to the
class $\BS_{\a}=\BS_{\a}(\beta) = \BS_{\a}(\beta,\,\hat{\Bth})$,\ $\a\in\R$, $0<\beta\leq1$, if
for any $l\ge 0$ and any non-negative $s\in\Z$ the condition
\begin{equation}\label{1b1:eq}
\1 b \1^{(\a)}_{l, s} :=
\max_{|\bs| \le s}
 \sum\limits_{\bth\in\hat{\Bth}}\lu \bth\ru^{l}\sup_{\bxi}\,\lu\bxi\ru^{(-\a + |\bs|)\beta}
|\BD_{\bxi}^{\bs} \hat b(\bth, \bxi)|<\infty, \ \ |\bs| = s_1+ s_2 + \dots + s_d,
\end{equation}
is fulfilled.
The quantities \eqref{1b1:eq} define norms on the class $\BS_\a$. Note that $\BS_\a$ is an increasing function of $\a$,
i.e. $\BS_{\a}\subset\BS_{\g}$ for $\a < \g$. We also have:
\bee
\sum\limits_{\bth\in\hat{\Bth}}\lu\bth\ru^{l}\sup_{\bxi}\,\lu\bxi\ru^{(-\a+s+1)\b}(|\BD^{\bs}_{\bxi}\hat b(\bth, \bxi+ \boldeta) -
\BD^{\bs}_{\bxi}\hat b(\bth, \bxi)|)\le C\1 b\1^{(\a)}_{l, s+1}  \lu\boldeta\ru^{|\a-s-1|\b}
|\boldeta|, \ s = |\bs|, \label{differ:eq}
\ene
with a constant $C$ depending only on $\a, s$. For a vector $\boldeta\in\R^d$ introduce
the symbol
\begin{equation}\label{bboldeta:eq}
b_{\boldeta}(\bx, \bxi) = b(\bx, \bxi+\boldeta), \boldeta\in\R^d,
\end{equation}
so that $\hat b_{\boldeta}(\bth, \bxi) = \hat b(\bth, \bxi+\boldeta)$ .
The bound \eqref{differ:eq} implies that for all $|\boldeta|\le C$ we have
\begin{equation}\label{differ1:eq}
\1 b - b_{\boldeta}\1^{(\a-1)}_{l, s}\le C \1 b\1^{(\a)}_{l, s+1}|\boldeta|,\
\end{equation}
uniformly in $\boldeta$: $|\boldeta|\le C$.

Now we define the PDO $\op(b)$ in the usual way:
\begin{equation}\label{eq:deff}
\op(b)u(\bx) = \frac{1}{(2\pi)^{\frac{d}{2}}} \int  b(\bx, \bxi)
e^{i\bxi \bx} (\mathcal Fu)(\bxi) d\bxi,
\end{equation}
the integrals being over $\Rd$. Under the condition $b\in\BS_\a$
 the integral in the r.h.s. is clearly finite for any
$u$ from the Schwarz
class $\plainS(\Rd)$. Moreover, the condition $b\in \BS_0$
guarantees the boundedness of $\op(b)$ in $L_2(\Rd)$, see
Proposition \ref{bound:prop}. Unless otherwise stated, from now on
$\plainS(\Rd)$ is taken as a natural domain for all PDO's at hand, when they act in $L_2$.
Notice that the operator $\op(b)$ is
symmetric if its symbol satisfies the condition
\begin{equation}\label{selfadj:eq}
\hat b(\bth, \bxi) = \overline{\hat b(-\bth, \bxi+\bth)}.
\end{equation}
We shall call such symbols \textsl{symmetric}.


We note that in the very beginning when we consider \eqref{eq:Sch}, our operator $\op(b)$ is a multiplication by a function $b$ (in particular, $b\in\BS_0$). However, during modifications and transformations below our perturbation will eventually become a pseudo-differential operator. 
Thus, it is convenient in abstract statements to consider $b$ a pseudo-differential symbol from some $\BS_\a$ class.

Now we list some properties of quasi-periodic PDO's. The proof is very similar (with obvious changes) to the proof of analogous statements in \cite{Sob}. In what follows, {\it if we need to calculate a product of two (or more) operators with some symbols $b_j\in\BS_{\a_j}({\Bth}_j)$ we will always consider that $b_j\in\BS_{\a_j}(\sum_j{\Bth}_j)$ where, of course, all added terms are assumed to have zero coefficients in front of them}.


\begin{prop}\label{bound:prop}
Suppose that $\1 b\1^{(0)}_{0, 0}<\infty$. Then $\op(b)$ is bounded in both $L_2(\R^d)$ and $B_2(\R^d)$ and $\|\op(b)\|\le \1 b \1^{(0)}_{0, 0}$.
\end{prop}

Since $\op(b) u\in\plainS(\Rd)$ for any $b\in\BS_{\a}$ and $u\in
\plainS(\Rd)$, the product $\op(b) \op(g)$, $b\in \BS_{\a}(\hat{\Bth}_1), g\in
\BS_{\g}(\hat{\Bth}_2)$, is well defined on $\plainS(\Rd)$. A straightforward
calculation leads to the following formula for the symbol $b\circ g
$ of the product $\op(b)\op(g)$:
\begin{equation*}
(b\circ g)(\bx, \bxi) = \sum_{\bth\in{\Bth}_1,\, \bphi\in{\Bth}_2}
\hat b(\bth, \bxi +\bphi) \hat g(\bphi, \bxi) e^{i(\bth+\bphi)\bx},
\end{equation*}
and hence
\begin{equation}\label{prodsymb:eq}
\widehat{(b\circ g)}(\bchi, \bxi) =
\sum_{\bth +\bphi = \bchi} \hat b (\bth, \bxi +\bphi) \hat g(\bphi,
\bxi),\ \bchi\in{\Bth}_1+{\Bth}_2,\ \bxi\in \Rd.
\end{equation}

We have

\begin{prop}\label{product:prop}
Let $b\in\BS_{\a}({\Bth}_1)$,\ $g\in\BS_{\g}({\Bth}_2)$. Then $b\circ g\in\BS_{\a+\g}({\Bth}_1+{\Bth}_2)$
and
\begin{equation*}
\1 b\circ g\1^{(\a+\g)}_{l,s} \le C \1 b\1^{(\a)}_{l,s} \1 g\1^{(\g)}_{l+(|\a|+s)\beta,s},
\end{equation*}
with a constant $C$ depending only on $l,\ \alpha,\ s$.
\end{prop}

\subsection{Gauge transform and the symbol of the resulting operator}
From now on we fix $\b:\ 0<\beta<\alpha_1$. The symbols we are going to construct will depend on $\rho_n$; this dependence will usually be omitted from the notation.

Let $\iota\in \plainC\infty(\R)$ be a non-negative function such
that
\begin{equation}\label{eta:eq}
0\le\iota\le 1,\ \ \iota(z) =
\begin{cases}
& 1,\  z \le \frac{1}{4};\\
& 0,\  z \ge \frac{1.1}{4}.
\end{cases}
\end{equation}
For $\bth\in \Bth, \bth\not = \mathbf 0$,
define the following $\plainC\infty$-cut-off functions:
\begin{equation}\label{el:eq}
\begin{cases}
e_{\bth}(\bxi) =&\ \iota\biggl({\biggl|\dfrac{|\bxi
+\bth/2|-3\rho_n}{10\rho_n}\biggr|}
 \biggr),\\[0.5cm]
\varphi_{\t}(\bxi) = &\  1 - \iota\biggl(\dfrac{|\lu\bth,\bxi + \bth/2\ru|}
{\rho_n^\beta |\bth|}\biggr).
\end{cases}
\end{equation}
The function $e_{\bth}$ is supported
in the shell $\rho_n/4\le |\bxi+\bth/2|\le 23\rho_n/4$. We point out that
\begin{equation}\label{symmetry:eq}
e_{\bth}(\bxi) = e_{-\bth}(\bxi + \bth),\ \ \  
\varphi_{\bth}(\bxi) =  \varphi_{-\bth}(\bxi + \bth).
\end{equation}
Note that the above functions satisfy the estimates
\begin{equation}\label{varphi:eq}
|\BD^{\bs}_{\bxi} e_{\bth}(\bxi)|
+
|\BD^{\bs}_{\bxi}\varphi_{\bth}(\bxi)|  \ll \rho_n^{-\beta |\bs|}.
\end{equation}
As before, we assume that $\tilde k$ is fixed. Put
$$
\tilde{\chi}_{\bth}(\bxi):=e_{\bth}(\bxi)\varphi_{\bth}(\bxi)(|\bxi+\bth|^2-|\bxi|^2)^{-1}=\frac{e_{\bth}(\bxi)\varphi_{\bth}(\bxi)}
{2\lu\bth,\bxi+\frac{\bth}{2}\ru}
$$
when $\bth\not={\bf 0}$, and $\tilde{\chi}_{\bf 0}(\bxi)=0$.

\bet\label{kandibober}
We can find a unitary operator $U$ and self-adjoint operators $H_1$ and $H_2$ such that the following properties hold:

1. $H_1=U^{-1}HU$;

2. $||H_1-H_2||\le \rho_n^{-2\beta\tilde k}$;

3. $H_2=-\Delta+W$, where $W=W_{\tilde k}$ is the operator with symbol
$w=w_{\tilde k}(\bx,\bxi)$ and $w$ satisfies the following property:

\bee\label{eq:b3}
\hat w(\bth,\bxi)=0, \ \mathrm{if} \
(\bxi\not\in\La(\bth)\ \&\ \bxi\in\CA), \  \mathrm{or} \ (\bxi+\bth\not\in\La(\bth)\ \&\ \bxi\in\CA),  \  \mathrm{or} \
(\bth\not\in\Bth_{\tilde k});
\ene

4. $U=e^{i\Psi}$, where $\Psi = \sum_{j=1}^{\tilde k} \Psi_j,\ \Psi_j = \op(\psi_j)$. Moreover, $\psi_j,\, b_j,\,t_j\in\BS_\gamma(\beta)$ for any $\gamma\in\R$, $\Psi$ is a bounded operator, and

\bee\label{estpsi}
\1\psi_j\1^{(\gamma)}_{l,s}\leq C_j\rho_n^{\beta(1-\gamma-2j)}\left(\1 b\1^{(0)}_{l_j,s_j}\right)^j,\ \ j\geq 1.
\ene
Assuming $\rho_0$ is large enough (depending on $l,s,\gamma,b$ and ${\tilde k}$), we get
\bee\label{estpsitotal}
\1\psi\1^{(\gamma)}_{l,s}\ll\rho_n^{-\beta(1+\gamma)}\1 b\1^{(0)}_{l,s};
\ene

5. The symbol of $W$ satisfies

\bee
\hat w_{\tilde k}(\bth,\bxi)=\hat y_{\tilde k}(\bth,\bxi)(1-e_{\bth}(\bxi)\varphi_{\bth}(\bxi)),
\ene
where $\hat{y}_{\tilde k}(\bth,\bxi)=0$ for $\bth\not\in\Bth_{\tilde k}$. Otherwise,
\begin{equation}\label{symbol}
\begin{split}
&\hat{y}_{\tilde k}(\bth,\bxi)=\hat{b}(\bth)+\sum\limits_{s=1}^{{\tilde k}-1}\sum  C_s(\bth,\bxi)\hat{b}(\bth_{s+1})\prod\limits_{j=1}^s \hat{b}(\bth_j)\tilde{\chi}_{\bth_j'}(\bxi+\bphi_j')\cr
&=\hat{b}(\bth)+\sum\limits_{s=1}^{{\tilde k}-1}\sum
C_s(\bth,\bxi)\hat{b}(\bth_{s+1})\prod\limits_{j=1}^s \hat{b}(\bth_j)\frac{e_{\bth_j'}(\bxi+\bphi_j')\varphi_{\bth_j'}(\bxi+\bphi_j')}
{2\lu\bth_j',\bxi+\bphi_j'+\frac{\bth_j'}{2}\ru},
\end{split}
\end{equation}
where the second sums are taken over all $\bth_j\in\Bth$,
$\bth_j',\bphi_j'\in\Bth_{s+1}$ and
\begin{equation}\label{indsconst}
C_s(\bth,\bxi)=\sum\limits_{p=0}^s\sum\limits_{\bth_j '',\bphi_j
''\in\Bth_{s+1}\ (1\leq j\leq p)}C_s^{(p)}(\bth) \prod\limits_{j=1}^p
e_{\bth_j ''}(\bxi+\bphi_j '')\varphi_{\bth_j ''}(\bxi+\bphi_j '').
\end{equation}
Here $C_s^{(p)}(\bth)$ depend on $s,\ p$ and all vectors
$\bth,\bth_j,\bth_j',\bphi_j',\bth_j '',\bphi_j ''$. At the
same time, coefficients $C_s^{(p)}(\bth)$ can be bounded uniformly
by a constant which depends on $s$ only. We apply the convention that $0/0=0$. 

\ent

The next results only partially have been proved explicitly in \cite{PaSh1}, but follow easily from the previous Theorem.

\bet\label{kandibober1}

1. We have $\hat{\psi}_m(\bth,\bxi)=0$ for $\bth\not\in\Bth_m$. Otherwise,
\begin{equation}\label{inds1}
\hat{\psi}_m(\bth,\bxi)=\sum C_m'(\bth,\bxi)\prod\limits_{j=1}^m \hat{b}(\bth_j)\tilde{\chi}_{\bth_j'}(\bxi+\bphi_j'),
\end{equation}
where the sum is taken over all $\bth_j\in\Bth$,
$\bth_j',\bphi_j'\in\Bth_m$ and $C_m'(\bth,\bxi)$ admit representation similar to \eqref{indsconst}.

2. We have 
\begin{equation}\label{UU}
U=I+\sum\limits_{m=1}^{\tilde k}\Psi'_m + U_{\tilde k+1},\ \ \ \ \ \Psi'_m={\rm Op}(\psi'_m),
\end{equation}
where symbols $\psi'_m$ admit representation similar to \eqref{inds1} and the estimate slightly worse than \eqref{estpsi}:
\bee\label{esttildepsi}
\1\psi'_m\1^{(\gamma)}_{l,s}\leq C_m\rho_n^{-\beta(m+\gamma)}
\left(\1 b\1^{(0)}_{l_m,s_m}\right)^m,\ \ m\geq 1. 
\ene
The error term $U_{\tilde k+1}$ admits the estimate
\bee\label{esterrorU}
\1 u_{{\tilde k}+1}\1^{(\gamma)}_{l,s}\ll \rho_n^{-\beta({\tilde k}+1+\gamma)}
\left(\1 b\1^{(0)}_{l_{{\tilde k}+1},s_{{\tilde k}+1}}\right)^{{\tilde k}+1}.
\ene
\ent
\bep
The first statement was proven in \cite{PaSh1}. The second statement follows by expanding the exponential $U=e^{i\Psi}$ into the Taylor series and estimating each term using \eqref{estpsi} and Proposition~\ref{product:prop}.
\enp
\bel\label{cor:4.1}
Fix $s\leq\beta\tilde k/4$. Then
\bee\label{4.1raz}
\|(H_1-H_2)(-\Delta+1)^s\|\leq \rho_n^{-\beta\tilde k}.
\ene
As a consequence, the estimate \eqref{4.1} holds.
\enl
\bep
Though the estimate \eqref{4.1raz} was not written in \cite{PaSh1} explicitly, it easily follows from the proofs of Theorems \ref{kandibober} and \ref{kandibober1} there.
\enp
\bec\label{cor:Udelta} 
As a function of $\bxi$, $\CF (U\de_{\bx_0})$ is uniformly bounded 
over $\bx_0$ (and $\bxi$), with constant depending only on $\tilde k$. 
Moreover, we have
\begin{equation}\label{fourierU}
\begin{split}
\CF(U\delta_{\bx_0})(\bxi)=(2\pi)^{-d/2}e^{-i\langle\bxi,\bx_0\rangle}+(2\pi)^{-d/2}\sum\limits_{m=1}^{\tilde k}\sum\limits_{\bth\in\Bth_m}\hat{\psi'}_m(\bth,\bxi-\bth)e^{-i\langle\bxi-\bth,\bx_0\rangle}+R(\bxi;\bx_0)=:\cr  
u(\bxi;\bx_0)+R_{\bx_0},
\end{split}
\end{equation}
where 
\begin{equation}\label{errnorm}
\|R_{\bx_0}\|
\leq \1 u_{{\tilde k}+1}\1^{(-\frac{d+1}{2\beta})}_{0,0}\ll 
\rho_n^{-\beta({\tilde k}+1)+\frac{d+1}{2}}\left(\1 b\1^{(0)}_{l_{{\tilde k}+1},s_{{\tilde k}+1}}\right)^{{\tilde k}+1}.
\end{equation}
uniformly over $\bx_0$.
\enc
\bep
Use \eqref{esterrorU} with $\gamma=-(d+1)/(2\beta)$, $l=s=0$.
\enp

\section{Proof of the main results}

Now we carry on with the proof using the gauge operators from the previous section. 
For any set $\CC\subset\R^d$ by $\CP^B(\CC)$ we denote  the orthogonal projection onto
$\mathrm{span}\{\be_{\bxi}\}_{\bxi\in\CC}$ in $B_2(\R^d)$ and by $\CP(\CC)=\CP^L(\CC)$ the same projection in $L_2(\R^d)$, i.e.
\bee
\CP(\CC)=\CF^*\chi_{\CC}\CF,
\ene
where $\CF$ is the Fourier transform and $\chi_{\CC}$ is the operator of multiplication
by the characteristic function of $\CC$.
Obviously, $\CP(\CC)$ is a well-defined (resp. non-zero) projection iff
$\CC$ is measurable (resp. has non-zero measure). 

Let $N$ be fixed (this is the number of precise asymptotic terms we need to obtain) and let $\la=\rho^2$ with $\rho\in I_n$ and $n$ being fixed (and large).  
We start by using Theorem \ref{kandibober} and find operators $U$, $H_1$ and $H_2$ with the properties listed there and put $\ep:=\rho_n^{-\beta\tilde k}$. The value of $\tilde k$ will be chosen later; it will depend on $N$. We are interested in the value of the kernel of the spectral projection $e(\la;H;\bx_0,\by_0)$, when $\bx_0$ and $\by_0$ are fixed. We can write this value, at least formally, as
\bee\label{e}
e(\la;H;\bx_0,\by_0)=(E_{\la}(H)\de_{\bx_0},\de_{\by_0})=
(E_{\la}(H)(E_{\la}(H)\de_{\bx_0}),E_{\la}(H)\de_{\by_0}).
\ene
As we have seen in Lemma \ref{L2norm}, both $E_{\la}(H)\de_{\bx_0}$ and $E_{\la}(H)\de_{\by_0}$ are elements of $L_2$, so the inner product in the RHS of \eqref{e} can be understood in the usual $L_2$ sense. Since $H_1=UHU^{-1}$, we have $E_{\la}(H)=U^{-1}E_{\la}(H_1)U$, and so 
\bee\label{e1}
e(\la;H;\bx_0,\by_0)=(E_{\la}(H_1)U\de_{\bx_0},U\de_{\by_0}). 
\ene
Also, 
\bee
||E_{\la}(H_1)U\de_{\bx_0}||=||UE_{\la}(H)\de_{\bx_0}||=||E_{\la}(H)\de_{\bx_0}||,
\ene
so $E_{\la}(H_1)U\de_{\bx_0}$ belongs to $L_2$ and (see Lemma \ref{L2norm})
\bee\label{budet}
\|E_{\la}(H_1)U\de_{\bx_0}\|=O(\la^{d/4}).
\ene 
We apply Lemma \ref{lem:2} with $\de:=\ep^{1/2}$ and obtain
\bee\label{eq:e2}
\begin{split}& 
||E_{\la}(H_2)U\de_{\bx_0}-E_{\la}(H_1)U\de_{\bx_0}||\ll ||E([\la-\ep^{1/2},\la+\ep^{1/2}];H_2)U\de_{\bx_0}||+\cr & \ep^{1/2}||E((-\infty,\la];H_2)U\de_{\bx_0}||+
\ep^{1/2}||(H_2+(1-a)I)^{-s}U\de_{\bx_0}||
\end{split}
\ene
and
\bee\label{eq:e2'}
\begin{split}& 
||E_{\la}(H_2)U\de_{\bx_0}-E_{\la}(H_1)U\de_{\bx_0}||\ll ||E([\la-\ep^{1/2},\la+\ep^{1/2}];H_1)U\de_{\bx_0}||+\cr & \ep^{1/2}||E((-\infty,\la];H_1)U\de_{\bx_0}||+
\ep^{1/2}||(H_1+(1-a)I)^{-s}U\de_{\bx_0}||.
\end{split}
\ene
Now, we choose $s>d/4$ so that (see Corollary \ref{cor:Udelta}) $(-\De+I)^{-s}U\de_{\bx_0}$ belongs to $L_2$ and
\bee
||(-\De+I)^{-s}U\de_{\bx_0}||\ll 1.
\ene
Then clearly we have 
\bee
||(H_j+(1-a)I)^{-s}U\de_{\bx_0}||\ll 1,\ \ \ j=1,2.
\ene
Therefore, 
the last terms in \eqref{eq:e2} and \eqref{eq:e2'} are $O(\ep^{1/2})$. 
We also notice that \eqref{budet} and \eqref{eq:e2'} imply
\bee\label{budet1}
\|E_{\la}(H_2)U\de_{\bx_0}\|=O(\la^{d/4}).
\ene 
What all this means is that it is enough to obtain precise asymptotics for the spectral projection of $H_2$, and this is what we will concentrate on now. Moreover, instead of studying
$(E_{\la}(H_2)U\de_{\bx_0},U\de_{\by_0})$ for $\la\in I_n$, we will 
study 
\bee\label{eq:tildee}
e(\la',\la'';H_2;U;\bx_0,\by_0):=(E([\la',\la''];H_2)U\de_{\bx_0},U\de_{\by_0})
\ene
for $\la'=(\rho')^2$, $\la''=(\rho'')^2$, $\rho',\rho''\in I_n$ with $\la''\ge\la'$. As we will see later, the following estimate holds: 
\bee\label{nuzhno}
||E([\la-\ep^{1/2},\la+\ep^{1/2}];H_2)U\de_{\bx_0}||^2=e(\la-\ep^{1/2},\la+\ep^{1/2};H_2;U;\bx_0,\bx_0)\ll\ep^{1/2}\rho_n^d.
\ene
In fact, \eqref{nuzhno} is an immediate corollary of \eqref{alreadyhere}.
Now \eqref{nuzhno}, \eqref{budet1}, and \eqref{eq:e2} imply 
\bee\label{eq:e3}
||E_{\la}(H_2)U\de_{\bx_0}-E_{\la}(H_1)U\de_{\bx_0}||\ll \ep^{1/4}\rho_n^{d/2}.
\ene
By choosing $\tilde k>4(M+2d)/\beta$, we make sure that the RHS of \eqref{eq:e3} is $O(\rho_n^{-M-d})$ and, therefore,  
\bee\label{eq:e4}
\begin{split}& 
e(\la;H;\bx_0,\by_0)=(E_{\la}(H_1)U\de_{\bx_0},E_{\la}(H_1)U\de_{\by_0})=\cr & (E_{\la}(H_2)U\de_{\bx_0},E_{\la}(H_2)U\de_{\by_0})+
O(\rho_n^{-M}).  
\end{split}
\ene
Here we also used \eqref{budet} and \eqref{budet1}.

Therefore, from now on we discuss only the spectral projections of $H_2$. Condition \eqref{eq:b3} implies that
for each $\bxi\in\CA$ the subspace $\CP^B(\BUps(\bxi))\GH$ is an invariant
subspace of $H_2$ acting in $B_2(\R^d)$. When we consider the action of $H_2$ in $L_2(\R^d)$, this subspace becomes trivial, so in order to get an interesting invariant subspace in $L_2$, we need to integrate $\CP(\BUps(\bxi))\GH$ over $\bxi$ in some open domain. For example, the subspace $\CP^L(\CA)L_2(\R^d)$ is the invariant subspace in $L_2$. 
We denote the 
dimension of $\CP^B(\BUps(\bxi))\GH$ by $q=q(\bxi)$ (which is finite by Lemma \ref{lem:finiteBUps}) and put
\bee\label{eq:h3bxi}
H_2(\bxi)=H_2(\BUps(\bxi)):=H_2\bigm|_{\CP^B(\BUps(\bxi))\GH}.
\ene
This operator acts in a finite-dimensional space $\GH_{\bxi}:=\CP^B(\BUps(\bxi))\GH$, so
its spectrum is purely discrete; we denote its eigenvalues (counting
multiplicities) by
$\la_1(\BUps(\bxi))\le\la_2(\BUps(\bxi))\le\dots\le\la_q(\BUps(\bxi))$
and the corresponding orthonormalized eigenfunctions by
$\{h_{j,\BUps(\bxi)}(\bx)\}$. Next, we list all points
$\boldeta\in\BUps(\bxi)$ in increasing order of their absolute
values (and if the absolute values are equal, we label them in any reasonable way we want, for example, in the lexicographic order of their coordinates). In such a way, we have put into correspondence to each point
$\boldeta\in\BUps(\bxi)$ a natural number $t=t(\boldeta)$ so that
$t(\boldeta)< t(\boldeta')$ if $|\boldeta|< |\boldeta'|$. Now we define the
mapping $g:\CA\to\R$ which puts into correspondence to each point
$\boldeta\in\CA$ the number $\la_{t(\boldeta)}(\BUps(\boldeta))$.  Similarly, we define the mapping $h:\CA\to
B_2(\R^d)$ by the formula $h_{\bxi}:=h_{t(\bxi),\BUps(\bxi)}$. Then for each $\bxi\in\CA$ the expression
$(2\pi)^{-d}\sum_{\boldeta\in\BUps(\bxi)}h_{\boldeta}(\bx)\overline{h_{\boldeta}(\by)}$
is the integral kernel of the projection $\CP(\BUps(\bxi))$. 
When $\bxi\not\in\CA$, we put
$g(\bxi):=|\bxi|^2$ and $h_{\bxi}:=\be_{\bxi}$, so that
now the functions $g$ and $h$ are defined on all $\R^d$. Denote 
\bee\label{eq:mla} 
G_{\la}:=\{\bxi\in\R^d,\,g(\bxi)\le\la\}. 
\ene
It
has been shown in \cite{PaSh1} that $\{h_{\bxi}\}_{\bxi\in\R^d}$ is
an orthonormal basis in $B_2(\R^d)$. Moreover, for each $\lambda\in
[0.75\la_n,17\la_n]$ the function
\bee\label{eq:projection}
e(\la;H_2;\bx,\by):=(2\pi)^{-d}\int_{G_{\la}}h_{\bxi}(\bx)
\overline{h_{\bxi}(\by)}d\bxi,\,\,\bx,\by\in\R^d,
\ene
is the
integral kernel of the spectral projection $E_{\la}(H_2)$
of the operator $H_2$ in $L_2(\R^d)$. Formula \eqref{eq:projection} was used in \cite{PaSh1} to compute the IDS of the operator $H$. It will be more convenient for us to use \eqref{eq:tildee} and \eqref{spectral} below  to compute the spectral function of $H$. 

Since each space $\CP(\Bxi(\GV))\GH$ (and even each component 
$\CP(\Bxi(\GV)_p)\GH$)
is invariant under $H_2$, they are invariant for $E_{\la}(H_2)$ as well, and therefore, if we denote  
\bee
E(\GV;p;I;H_2):=\CP(\Bxi(\GV)_p)E(I;H_2)\CP(\Bxi(\GV)_p), 
\ene
we have
\bee\label{eq:tildee1}
\bes
e(\la',\la'';H_2;U;\bx_0,\by_0)&=\sum_{\GV\in\CV}\sum_p(E(\GV;p;[\la',\la''];H_2)U\de_{\bx_0},U\de_{\by_0})\\
&=:\sum_{\GV\in\CV}\sum_p e_{\GV;p}(\la',\la'';H_2;U;\bx_0,\by_0),
\end{split}
\ene
the sum being over all lattice subspaces of dimension at most $d-1$ (the contribution from $\Bxi(\R^d)$ is zero for large $\rho_n$ since
$\Bxi(\R^d)$ is bounded). 

Let us fix the connected component $\Bxi(\GV)_p$ of the resonance zone and study the contribution from it to the RHS of \eqref{eq:tildee1}. 
Suppose that two points $\bxi$ and $\bxi'$ have the
same coordinates $X$ and $\tilde\Phi$ and different coordinates $r$ and $r'$. Then (see \cite{PaSh1}) $\bxi\in\Bxi_p$ implies $\bxi'\in\Bxi_p$ and
$\BUps(\bxi')=\BUps(\bxi)+(\bxi'-\bxi)$. This shows that two spaces $\GH_r(X,\tilde\Phi):=\GH_{\bxi}$
and $\GH_{r'}(X,\tilde\Phi):=\GH_{\bxi'}$ have the same dimension and, moreover,
there is a natural isometry
$F_{\bxi,\bxi'}:\GH_{\bxi}\to\GH_{\bxi'}$ given by $F:\be_{\bnu}\mapsto\be_{\bnu+(\bxi'-\bxi)}$,
$\bnu\in\BUps(\bxi)$. This isometry allows us to `compare' operators acting in $\GH_{\bxi}$ and
$\GH_{\bxi'}$. Therefore, abusing slightly our notation, we can assume that $H_2(\bxi)$ and $H_2(\bxi')$
act in the same (finite dimensional) Hilbert space $\GH(X_1,\dots,X_m,\tilde\Phi_1,\dots,\tilde\Phi_{K})$.
We fix the values $(X_1,\dots,X_m,\tilde\Phi_1,\dots,\tilde\Phi_{K})$, $X=(X_1,\dots,X_m)\in\Om(\GV)$, $\tilde\Phi=(\tilde\Phi_1,\dots,\tilde\Phi_{K})\in M_p$  
and study how these
operators depend on $r$. We denote by $H_2(r)=H_2(r;X_1,\dots,X_m,\tilde\Phi_1,\dots,\tilde\Phi_{K})$ the operator
$H_2(\bxi)$ with $\bxi=(X_1,\dots,X_m,r,\tilde\Phi_1,\dots,\tilde\Phi_{K})$, acting in $\GH(X_1,\dots,X_m,\tilde\Phi_1,\dots,\tilde\Phi_{K})$. Sometimes we will still want to emphasise that those operators for each $r$ act in their own Hilbert spaces; then, we will use the notation $\tilde H_2(r)$. Thus, $H_2(r)$ acts in $\GH(X_1,\dots,X_m,\tilde\Phi_1,\dots,\tilde\Phi_{K})$ and $\tilde H_2(r)$ acts in $\GH_r(X_1,\dots,X_m,\tilde\Phi_1,\dots,\tilde\Phi_{K})$. 
We will also use the notation $\int\limits_{[a,b]}^{\oplus}\, \tilde H_2(r)\,dr$ for the direct integral of these operators (emphasising that each `fiber' operator $\tilde H_2(r)$ acts in its own Hilbert space 
$\GH_{r}$).

As we have seen in Theorem \ref{kandibober},
the symbol of the operator $H_2$ satisfies
\bee\label{eq:nn1}
h_2(\bx,\bxi)=|\bxi|^2+{w}_{\tilde k}(\bx,\bxi)=r^2+2r\lu\ba,\bn(\boldom)\ru+|\ba|^2+{w}_{\tilde k}(\bx,\bxi)+|X|^2,
\ene
where
the Fourier coefficients of $w_{\tilde k}$ satisfy 
\eqref{symbol}, \eqref{indsconst}, and we have denoted  
$\boldom=\bxi_{\GV^{\perp}}-\ba$. This immediately implies that
the operator $H_2(r)$ is monotone increasing in $r$; in particular, all its eigenvalues
$\la_j(H_2(r))$ are increasing in $r$. Thus, the function $g(\bxi)$ is an increasing function 
of $r$ if we fix other coordinates of $\bxi$, so the equation
\bee\label{eq:tau}
g(\bxi)=\rho^2
\ene
has a unique solution if we fix the values $(X_1,\dots,X_m,\tilde\Phi_1,\dots,\tilde\Phi_{K})$; we
denote the $r$-coordinate of this solution by $\tau=\tau(\rho)=\tau(\rho;X_1,\dots,X_m,\tilde\Phi_1,\dots,\tilde\Phi_{K})$,
so that
\bee\label{eq:tau1}
g(\bxi(X_1,\dots,X_m,\tau,\tilde\Phi_1,\dots,\tilde\Phi_{K}))=\rho^2.
\ene

 Notice that the operator $H_2(r)$ will stay the same if we replace the coordinates $X$ to $Y$ in such a way that the resulting point 
$\boldeta=(Y_1,\dots,Y_m,r,\tilde\Phi_1,\dots,\tilde\Phi_{K})$ is equivalent to $\bxi$ so that $\BUps(\bxi)=\BUps(\boldeta)$. In this case we will say that $X$ and $Y$ are equivalent, $X\leftrightarrow Y$. We also 
denote $\tilde\Omega(\GV):=\Omega(\GV)/\leftrightarrow$. If we fix the class of equivalence $\tilde X$ from $\tilde\Omega(\GV)$ and $\tilde\Phi$, 
there will be $q$ representatives $X^{(1)},\dots,X^{(q)}$ from the class $\tilde X$. Then we will have $q$ numbers $\tau(\rho;X^{(j)},\tilde\Phi)$, $j=1,\dots,q$ corresponding to a single pair $(\tilde X,\tilde\Phi)$; we will label them in the increasing order and denote $\tau_j(\rho;\tilde X,\tilde\Phi)$, $j=1,\dots,q$.

Let us denote by $S=S(r)$ the operator with symbol $2r\lu\ba,\bn(\boldeta)\ru+|\ba|^2+{w}_{\tilde k}(\bx,\bxi)+|X|^2$ acting in $\GH(X_1,\dots,X_m,\tilde\Phi_1,\dots,\tilde\Phi_{K})$, so that
$H_2(r)=r^2 I+S(r)$. Then \eqref{symbol} implies
\bee\label{eq:newS1}
||S(r)||\ll\rho_n^{1+\al_{d}},\ \ ||S'(r)||\ll \rho_n^{\al_{d}},
\ene
and
\bee\label{eq:newS2}
||\frac{d^l}{dr^l}S(r)||\ll \rho_n^{-l},\ \ l\ge 2.
\ene
The operator $H_2(r)$ can be analytically continued to the complex plane 
(at least to the domain $|z-\rho|\le \rho/8$; see Remark 10.1 from \cite{PaSh1}). We will denote such an extension by $H_2(z)$. 

Now consider the function 
$u(\bxi;\bx_0)$ 
introduced in \eqref{fourierU}. Let us fix the values of $(X_1,\dots,X_m,\tilde\Phi_1,\dots,\tilde\Phi_{K})$ and consider the restriction of
$u(\bxi;\bx_0)$ to 
$\BUps(r;X,\tilde\Phi)$
as a function of $r$; we call it then $\tilde u=\tilde u(r;\bx_0)\in \GH(X_1,\dots,X_m,\tilde\Phi_1,\dots,\tilde\Phi_{K})$. Formula \eqref{fourierU} shows that $\tilde u(r;\bx_0)$ admits the analytical extension into the complex plane; we will denote this extension by $\tilde u(z;\bx_0)$. 

From the above discussion (see in particular Corollary~\ref{cor:Udelta}), it follows that 
\bee\label{quadratic}
\bes
 &e_{\GV;p}(\la',\la'';H_2;U;\bx_0,\by_0)\\
 &=\int_{\tilde\Omega(\GV)} d\tilde X\int_{M_p}d\tilde\Phi\int\limits_a^b r^K\left(E([\lambda',\lambda''];H_2(r))\tilde u(r;\bx_0),\tilde u(r;\by_0)\right)dr+O(\rho_n^{-M}).
\end{split}  
\ene
Here,  $d\tilde\Phi$ is understood as \eqref{eq:surfaceelement} (see also \eqref{eq:etajdash}), $a:=\tau_1(\rho';\tilde X,\tilde\Phi)$ (in fact, we can decrease $a$ and the integral will not change) 
and $b:=\tau_q(\rho'';\tilde X,\tilde\Phi)$ (similarly, $b$ can be increased without changing the integral). 

Now, we notice that all functions $\tau_j$ introduced above belong to the space $T([\lambda',\lambda''])$  introduced in Definition \ref{Tau} (this  easily follows from \eqref{eq:newS1}). Let $\tilde\Gamma'$ be an appropriate contour around interval $[a,b]$, i.e. $\tilde\Gamma'$ be completely inside the domain of analyticity of the operator $H_2(r)$. 

\begin{lem}\label{spectralint}Let $f(z)$ and ${g(z)}$ be analytic so that $\tilde\Gamma'$ is also inside their domain of analyticity. Then the following formula holds
\begin{equation}\label{spectral}
\int\limits_a^b \left(E(H_2(r);[\lambda',\lambda''])f(r),g(r)\right)dr=\frac{1}{2\pi i}\int\limits_{\lambda'}^{\lambda''}\,d\mu\oint\limits_{\tilde\Gamma'}
\left((H_2(z)-\mu)^{-1}f(z),g(\bar z)\right)dz.
\end{equation}
\end{lem}
\begin{proof}
First, consider the case when the intervals $[\tau_j(\lambda'),\tau_j(\lambda'')]$ do not intersect for all $j$. It is allowed that several eigenvalues coincide identically and thus are represented by the same function  $\tau^{-1}_j(r)$. Let $\Gamma$ be a contour around $[\la',\la'']$ which is inside the domain of analyticity of all the functions involved. By $\tilde\Gamma_j$ we denote a contour around interval $[\tau_j(\lambda'),\tau_j(\lambda'')]$ so that $\tilde\Gamma_j$ do not intersect for all $j$.  Then we have
$$
\int\limits_{[a,b]}^{\oplus}\, E(\la;\tilde H_2(r))\,dr=\int\limits_{[a,b]}^{\oplus}\,\sum\limits_{j:\, \tau_j(\lambda)>r}E_j(\tilde H_2(r))\,dr=\sum\limits_j
\int\limits_{[a,\tau_j(\lambda)]}^{\oplus}
\,E_j(\tilde H_2(r))\,dr.
$$
Here, $E_j(\tilde H_2(r))$ is the orthogonal projector onto the span of $h_{j,\BUps(\bxi)}$ -- the eigenfunction of $\tilde H_2(r)$ corresponding to  $\la_j(\BUps(\bxi))$ -- its eigenvalue number $j$. Recall that  operators $\tilde H_2(r)$ and $E_j(\tilde H_2(r))$ act in the space $\GV_r$. 
Thus,
\bee
\bes
\int\limits_{[a,b]}^{\oplus}\,  E([\lambda',\lambda''];\tilde H_2(r))\,dr&=\sum\limits_j
\int\limits_{[\tau_j(\lambda'),\tau_j(\lambda'')]}^{\oplus}
\,E_j(\tilde H_2(r))\,dr\\
&=
-\frac{1}{2\pi i}\sum\limits_j\int\limits_{[\tau_j(\lambda'),\tau_j(\lambda'')]}^{\oplus}\,dr\oint\limits_\Gamma\frac{d\mu}{\tilde H_2(r)-\mu}.
\end{split}
\ene
 We use the singular decomposition inside $\tilde\Gamma_j$ for $r$ and inside $\Gamma$ for $\mu$:
\begin{equation*}
\left((H_2(r)-\mu)^{-1}f(r),g(r)\right)=\frac{\left(E_j(H_2(r))f(r),g(r)\right)}{\tau_j^{-1}(r)-\mu}+G_j(r,\mu)
\end{equation*}
with analytic function $G_j$. Now, Lemma~\ref{aux} leads to the following identity.
\begin{equation*}
\begin{split}
\int\limits_a^b \left(E(H_2(r);[\lambda',\lambda''])f(r),g(r)\right)dr=\frac{1}{2\pi i}\int\limits_{\lambda'}^{\lambda''}\,d\mu\oint
\limits_{\cup_j\tilde\Gamma_j}\left((H_2(z)-\mu)^{-1}f(z),g(\bar z)\right)dz
=\cr \frac{1}{2\pi i}\int\limits_{\lambda'}^{\lambda''}\,d\mu\oint
\limits_{\tilde\Gamma'}\left((H_2(z)-\mu)^{-1}f(z),g(\bar z)\right)dz.
\end{split}
\end{equation*}
Next, if  intervals $[\tau_j(\lambda'),\tau_j(\lambda'')]$ can touch each other at boundary points only then formula \eqref{spectral} can be easily justified by continuity from the intervals $[\lambda'+\epsilon,\lambda''-\epsilon]$ (note that only construction of the auxiliary contours $\tilde{\Gamma}_j$ was based on the non-intersection assumption). Finally, to prove \eqref{spectral} in the general case it is enough to divide interval $[\lambda',\lambda'']$ into finitely many subintervals at pre-images of the points of intersection of different functions $\tau_j$ and combine the results for each subinterval (note that for each  subinterval the contour $\tilde\Gamma'$ can, by analyticity, be extended  to the contour around $[a,b]$).
\end{proof}

Now we want to extend the contour of integration in the RHS of \eqref{spectral}. Denote by $\hat\Gamma$ the contour of radius $\rho_n/8$ around $\rho'=(\la')^{1/2}$. The following result was proved in \cite{MoPaSh}
\bel
Let $0<\rho''-\rho'<\rho_n/8$. Assume that the operator $(H(z)-\mu)$ is not invertible, where $\mu\in[\la',\la'']$ and $z$ is inside $\hat\Gamma$. 
Then $z$ is real (and belongs to $[a,b]$, where $a=\tau_1(\rho')$ and $b=\tau_q(\rho'')$).  
\enl
\bec We have:
\bee\label{quadratic1}
\bes
 &e_{\GV;p}(\la',\la'';H_2;\bx_0,\by_0)\\
 &=\frac{1}{2\pi i}\int_{\tilde\Omega(\GV)} d\tilde X\int_{M_p}d\tilde\Phi
 \int\limits_{\lambda'}^{\lambda''}\,d\mu\oint\limits_{\hat\Gamma}
z^K \left((H_2(z)-\mu)^{-1}\tilde u(z;\bx_0),\tilde u(\bar z;\by_0)\right)dz+O(\rho_n^{-M}).
\end{split}
\ene
\enc
Now we carry on the calculations and notice that on the contour $\hat\Gamma$ we can extend the resolvent in the geometric series (recall that $H_2(z)=z^2I+S(z)$; we also denote $\mu=\rho^2$):
\begin{equation}\label{resolventS}
(H_2(z)-\mu)^{-1}=\sum\limits_{l=0}^\infty(-1)^l S^l(z)(z^2-\mu)^{-(l+1)},\ \ \ z\in\hat\Gamma.
\end{equation}
Therefore, using the Cauchy integral formula we have 
\bee\label{eq:Cauchy}
\bes
&\frac{1}{2\pi i}\oint\limits_{\hat\Gamma}
z^K \left((H_2(z)-\mu)^{-1}\tilde u(z;\bx_0),\tilde u(\bar z;\by_0)\right)dz\\
&= \frac{1}{2\pi i}\sum\limits_{l=0}^\infty
\oint\limits_{\hat\Gamma}
z^K(-1)^l(z-\rho)^{-(l+1)}(z+\rho)^{-(l+1)}\left( S^l(z)
\tilde u(z;\bx_0),\tilde u(\bar z;\by_0)\right)dz\\
&=\sum\limits_{l=0}^\infty
\frac{(-1)^l}{l!}\frac{d^l}{dr^l}[r^K(r+\rho)^{-(l+1)}\left( S^l(r)
\tilde u(r;\bx_0),\tilde u(r;\by_0)\right)]\bigm|_{r=\rho}.
\end{split}
\ene
Denote the RHS of \eqref{eq:Cauchy} by $\tilde J=\tilde J(\rho;\tilde X;\tilde\Phi;\bx_0,\by_0)$. 
Using explicit formulae \eqref{symbol} and \eqref{fourierU}, we deduce that $\tilde J$ is a sum of the terms of the following form:
\bee\label{eq:form}
\sum_{\boldeta\in\BUps(\bxi)} \sum_{\bnu\in\BUps(\bxi)}e^{-i\langle\boldeta,\bx_0\rangle}e^{i\langle\bnu,\by_0\rangle}\rho^{m}\tilde  f_1(X(\boldeta),X(\bnu))f_2(\Phi)\tilde f_3(X(\boldeta),X(\bnu);
\rho;\Phi),
\ene
$m\leq K-1\leq d-2$. Here, function $\tilde f_1$ is a smooth function of $X$ 
coordinates of $\boldeta$ and $\bnu$ only. It consists of contributions from the cut-off functions $\varphi_{\bth}$ with $\bth\in\GV$ and from the terms in \eqref{symbol}, \eqref{indsconst} corresponding
to $\bth'_j, \bth_j''\in\GV$. It also takes care of the inner products  
$\left( S^l(\rho)\tilde u(\rho;\bx_0),\tilde u(\rho;\by_0)\right)$; these inner products are the reason why we need double sum in \eqref{eq:form}: we have expressed this inner product as a sum of all matrix elements of  
$S^l(\rho)$ times the corresponding elements of the two vectors. 
The function $f_2(\Phi)$ is a product of powers of $\{\sin\Phi_q\}$. This
function 
comes from differentiating
\eqref{eq:innerproduct1} and $e^{-i\langle\bxi,\bx_0-\by_0\rangle}$ with respect to $r$. Finally, $\tilde f_3$ is of the following form:
\bee\label{eq:f3}
\tilde f_3(X;t;\Phi)=\prod_{t=1}^T
(l_t
+\rho\sum_{q} b_q^t\sin(\Phi_q))^{-k_t
}.
\ene
This function corresponds to the negative powers of inner products $\lu\bxi,\bth_t\ru$ given by
Lemma \ref{lem:products}, part (ii). Here, $\{b_q^t\}$ are coefficients in the decomposition $(\bth_t)_{\GV^{\perp}}=\sum_q b_q^t\tilde \bmu_q$;
recall that these numbers are all of the same sign and satisfy \eqref{eq:n10}. Without loss of
generality we will assume that all $b_q^t$ are non-negative. 
The number $l_t=l(b_1^t,\dots,b_{K+1}^t):=\lu X,(\bth_t)_{\GV}\ru+L_{m+1}\sum_{q}b_q^t$ satisfies
$l_t\asymp  \rho_n^{\al_{m+1}}$, since our
assumptions imply $|\lu X,\bth_{\GV}\ru|\ll\rho_n^{\al_m}$. This number depends on $X$,
but not on $\Phi$ or $\rho$. The number $k_t=k(b_1^t,\dots,b_{K+1}^t)$ is positive, integer, and independent of $\bxi$. 

By denoting $\bth:=\boldeta-\bnu\in (\BUps(\bxi)-\BUps(\bxi))$, we can re-write \eqref{eq:form} like this:
\bee\label{eq:form3}
\sum_{\bth\in(\BUps(\bxi)-\BUps(\bxi))}\sum_{\boldeta\in\BUps(\bxi)} e^{-i\langle\boldeta,\bx_0-\by_0\rangle}e^{-i\langle\bth,\by_0\rangle}\rho^{m}\tilde  f_1(X(\boldeta),X(\boldeta-\bth))f_2(\Phi)\tilde f_3(X(\boldeta),X(\boldeta-\bth);
\rho;\Phi).
\ene

Recall that \eqref{quadratic1} involves the integration against $d\tilde X$ (where $\tilde X$ is a class of equivalence with respect to $\leftrightarrow$). If we integrate against $dX$ instead, we can get rid of the summation over different $\boldeta$ in \eqref{eq:form3}: 
\bee\label{quadratic2}
\bes
 &e_{\GV;p}(\la',\la'';H_2;U;\bx_0,\by_0)\\
 &=\int_{\Omega(\GV)} dX\int_{M_p}d\tilde\Phi
 \int\limits_{\lambda'}^{\lambda''}\,\rho d\rho J(\rho;\tilde X;\tilde\Phi;\bx_0,\by_0)+O(\rho_n^{-M}),
\end{split}
\ene
where $J$ is the sum of the terms of the following form:
\bee\label{eq:form2} 
\sum_{\bth\in(\BUps(\bxi)-\BUps(\bxi))} e^{-i\langle\bxi,\bx_0-\by_0\rangle}\rho^{m}f_1(X(\bxi);\bth)f_2(\Phi)f_3(X(\bxi);\rho;\Phi),\ \ \ m\leq K-1\leq d-2.
\ene
with $f_j$ satisfying the same properties as before. 
Note that $\bxi$ is the point with coordinates $(X,\rho,\tilde\Phi)$. 

We have 
\bee\label{eq:in2'}
\langle\bxi,\bx_0-\by_0\rangle=\langle X, (\bx_0-\by_0)_\GV\rangle+\langle \ba, (\bx_0-\by_0)_{\GV^\perp}\rangle+ \rho\sum_q s_q\sin\Phi_q,
\ene
where constants $s_q$ are coefficients in the decomposition $(\bx_0-\by_0)_{\GV^\perp}=\sum_{q}s_q\tilde\bmu_q(p)$.

Our objective is to compute the sum of the integrals
of \eqref{eq:form2} over various domains $\Bxi(\GV)_p$.  So, we need to integrate the functions of the form
\begin{equation}\label{eq:in2}
F_K:=\frac{(\sin\Phi_1)^{n_1}\dots (\sin\Phi_K)^{n_K}(\sin\Phi_{K+1})^{n_{K+1}}}{\prod_{t=1}^T (l_t+\rho\sum_{j=1}^{K+1} b_j^t \sin\Phi_j)^{k_t}}
e^{-i\rho\sum_q s_q\sin\Phi_q}.
\end{equation}

\subsection{The diagonal case}
On the diagonal $\bx_0=\by_0$ the integral $\hat J_K:=\int\limits_{M_p}F_Kd\tilde \Phi$ is exactly of the same form as in \cite{PaSh1} and we can directly apply the results from \cite{PaSh1} , Lemma 10.4:
\bel\label{lem:integral1}
We have:
\begin{equation}\label{eq:JK3}
\hat J_K=
\sum_{q=0}^K(\ln({\rho}))^q\sum_{p=0}^\infty
d(p,q;n){\rho}^{-p},
\ene
where
\bee\label{eq:estimatee1}
|d(p,q;n)|\ll\rho_n^{2p/3}\rho_n^{-Q\beta},
\end{equation}
where $Q:=\sum_t k_t$. These estimates are uniform in $\CA_n$.\footnote{We recall that coefficients $d(p,q;n)$ depend on $X$ via $l_t$.}

\enl
From Lemma~\ref{lem:integral1} we immediately obtain (cf. \cite{PaSh1}) the asymptotic formula
\begin{equation}\label{almosthere}
\int_{\Omega(\GV)} dX\int_{M_p}d\tilde\Phi\,\rho J(\rho;\tilde X;\tilde\Phi;\bx_0,\by_0)=\sum_{q=0}^{d-1}(\ln({\rho}))^q\sum_{p=-d+1}^\infty
\hat d(p,q;n){\rho}^{-p},\ \ \ \rho>\rho_n^{2/3},
\end{equation}
where
\bee\label{almosthere1}
|\hat d(p,q;n)|\ll\rho_n^{2p/3+2(d-1)/3}.
\end{equation}
Now, to calculate \eqref{quadratic2} we integrate \eqref{almosthere} against $d\rho$. Then taking the summation over all $M_p$ and $\GV\in\CV$ (see \eqref{eq:tildee1}) we obtain
\begin{equation}\label{alreadyhere}
\bes
&e(\la'',\la';H_2;U;\bx_0,\bx_0)\\
&=
\sum_{q=0}^{d}(\ln({\lambda''}))^q\sum_{p=-d}^\infty
\tilde d(p,q;n){(\lambda'')}^{-p/2}-\sum_{q=0}^{d}(\ln({\lambda'}))^q\sum_{p=-d}^\infty
\tilde d(p,q;n){(\lambda')}^{-p/2},
\end{split}
\end{equation}
where
\bee\label{alreadyhere1}
|\tilde d(p,q;n)|\ll\rho_n^{2p/3+2d/3}.
\end{equation}

Recall that for $[\lambda',\lambda'']\subset I_n$ formula \eqref{alreadyhere} describes (up to an error $O(\rho_n^{-M})$) the kernel  $e([\lambda',\lambda''];H_2;U;\bx_0,\bx_0)$ introduced in \eqref{eq:tildee}. Strictly speaking, so far we have proved \eqref{alreadyhere} only assuming that $\rho''-\rho'<\rho_n/8$, but in the general case we cover the interval $[\lambda',\lambda'']$ by several (at most 16) smaller intervals where this assumption is satisfied, and then sum formulae \eqref{alreadyhere} for each of these intervals. 

Therefore, for $\lambda^{1/2}\in I_n$ we have 
\begin{equation}\label{I_nasymp}
\begin{split}&
e([-\infty,\lambda];H_2;U;\bx_0,\bx_0)=
e([\rho_n^2,\lambda];H_2;U;\bx_0,\bx_0)+e([-\infty,\rho_n^2];H_2;U;\bx_0,\bx_0)=\cr & 
\sum_{q=0}^{d}(\ln({\lambda}))^q\sum_{p=-d}^\infty
\tilde d(p,q;n){\lambda}^{-p/2}+c_n+O(\rho_n^{-M}),\ \ \ \lambda^{1/2}\in I_n,
\end{split}
\end{equation}
where 
\bee\label{eq:cn}
c_n:=e([-\infty,\rho_n^2];H_2;U;\bx_0,\bx_0)-\sum_{q=0}^{d}(\ln({\rho_n^2}))^q\sum_{p=-d}^\infty
\tilde a(p,q;n){\rho_n}^{-p}.
\ene
Obviously, $|c_n|<<\rho_n^{d+1}$. Using \eqref{eq:e4}, we obtain the following asymptotic expansion for $e([-\infty,\lambda];H;\bx_0,\bx_0)$:
\begin{equation}\label{I_nasymp1}
\begin{split}&
e([-\infty,\lambda];H;\bx_0,\bx_0)=
\sum_{q=0}^{d}(\ln({\lambda}))^q\sum_{p=-d}^{\infty}
a(p,q;n){\lambda}^{-p/2}+O(\rho_n^{-M})=\cr & 
\sum_{q=0}^{d}(\ln({\lambda}))^q\sum_{p=-d}^{6M}
a(p,q;n){\lambda}^{-p/2}+O(\rho_n^{-M}),\ \ \ \lambda^{1/2}\in I_n,
\end{split}
\end{equation}
where
\bee\label{I_nasymp2}
|a(p,q;n)(\bx_0)|\ll\rho_n^{2p/3+d+1},
\end{equation}
and $O$-term here does not depend on $n$ (though it can depend on $M$). In fact, $a(0,0;n)=\tilde d(0,0;n)+c_n$ and $a(p,q;n)=\tilde d(p,q;n)$ for $|p|+q\not=0$. It is also not hard to see that the estimates are uniform in $\bx_0\in\R^d$. This proves Lemma~\ref{main_lem}.
\ber
The explicit form of the principal term in \eqref{eq:main_lem1} is just a consequence of the a priori fact that $N(\rho^2;\bx)=C_d\rho^d(1+o(1))$ (see, e.g., \cite{AgKa}). 
\enr

\subsection{The off-diagonal case}

Now we assume that $\bx_0\ne\by_0$. Denote 
$\bn:=\frac{\bx_0-\by_0}{|\bx_0-\by_0|}$. 
As we have stated in the introduction, we will consider only the (generic) case when $\bx_0-\by_0$ is not orthogonal to any of the vectors $\bth\in \Bth_{\infty}$, so 
in particular $\bx_0-\by_0$ is not orthogonal to any of the vectors  $\bth\in \Bth_{\tilde k}$. For large $\rho$ this implies that the vector 
$\rho\bn$ belongs to the non-resonant region $\Bxi(\GX)$. What we plan to do is, essentially, computing integrals of the form \eqref{eq:form3} using the stationary phase approach. All the computations are quite standard; however, since we are using rather special set of coordinates, we have to be careful when introducing the partition of unity, etc. Therefore, we  give some explanations of what we do here, but do not calculate everything in detail. 

We start by introducing the partition of unity dictated by our coordinates in $\CA$. We start by treating the resonant regions of the lowest rank, i.e. computing the integrals over $\Bxi(\GX)$. So, let us fix the component 
\bee
\Bxi(\GX)_p = \{\bxi\in\R^d,\ \ \&\ \forall j \ \lu\bxi,\tilde\bmu_j(p)\ru > L_{1}\}
\ene
for some vectors $\{\tilde\bmu_j(p)\}_{j=1}^{d-1}$. 
Next, for each vector $\bth=\tilde\bmu_j(p)$, $j=1,\dots,d-1$  
we introduce three functions 
$e^k_{\bth}=e^k_{\bth}(\bxi)$, $\bxi\in\CA$, $k=1,2,3$, with the following properties:

1. All functions $e^k_{\bth}$ are smooth and satisfy $0\leq e^k_{\bth}\leq 1$ and $\sum_{k=1}^3e^k_{\bth}=1$ in $\CA$.

2. $e^1_{\bth}$ depends only on the projection 
$P_{\bth}(\bxi):=\langle\bxi,\bth\rangle/|\bth|$; $e^1_{\bth}$ equals zero when $P_{\bth}(\bxi)>\rho_n^{\al_1}$ and $e^1_{\bth}$ equals one when $P_{\bth}(\bxi)<\frac12\rho_n^{\al_1}$.  

3. $e^3_{\bth}$ depends only on the angle $\Phi_j:=\frac{\pi}{2}-\phi(\bxi-\ba,\tilde\bmu_j(p))$; $e^3_{\bth}$ equals zero when $P_{\bth}(\bxi)<\rho_n^{\al_1}$ and $e^3_{\bth}$ equals one when $P_{\bth}(\bxi)>2\rho_n^{\al_1}$.

4. All partial derivatives of these functions satisfy 
\bee\label{eq:partial_estimates}
|\BD^\bs_{\bxi}e^k_{\bth}|\ll\rho_n^{-|\bs|\beta}.  
\ene

It is a simple (but quite tedious) exercise to check that such a partition exists. Now we consider the contribution to \eqref{quadratic2} when we multiply the integrand by $\prod_{j=1}^{d-1} e^3_{\bth}$. 

Case 1. $\rho\bn\not\in \Bxi(\GX)_p$. Then the contribution to the integral \eqref{quadratic2} equals (see \eqref{eq:in2})
\begin{equation}\label{twod1}
\int_{M_p} \frac{(\sin\Phi_1)^{n_1}\dots (\sin\Phi_K)^{n_K}(\sin\Phi_{K+1})^{n_{K+1}}\,d\tilde\Phi}{\prod_{t=1}^T (l_t+\rho\sum_{j=1}^{K+1} b_j^t \sin\Phi_j)^{k_t}}
e(\tilde\Phi) e^{\pm i\rho|\bx_0-\by_0|f(\tilde\Phi)}.
\end{equation}
Here, $f$ is smooth (in fact analytic) in $\tilde\Phi$, and  $\nabla f$ is never equal to zero inside $M_p$. The best way to see this is to use the original form of $f$ (cf. \eqref{eq:in2} and \eqref{eq:in2'}). The function $e$  is a smooth cut-off with support strictly inside $M_p$ satisfying 
\bee\label{eq:partial_estimates1}
|\BD^\bs_{\tilde\Phi}e|\ll\rho_n^{|\bs|(1-\beta)} 
\ene
(these estimates follow from \eqref{eq:partial_estimates} after rescaling). 
It is easy to see that each integration by parts (we integrate the exponential and differentiate other factors along $(\bx_0-\by_0)$) improves the estimate of the integral by at least $\rho_n^{-\beta}$. This shows that these integrals are $O(\rho_n^{-M})$.

Case 2. $\rho\bn\in \Bxi(\GX)_p$. Then the contribution to the integral \eqref{quadratic2} is again given by \eqref{twod1}, the only difference being that now there is a point $c\in M_p$ such that $\nabla f(c)=0$, and 
the Hessian of $f$ at $c$ is non-degenerated (again see \eqref{eq:in2} and \eqref{eq:in2'}). Here we need to be just a little bit more careful. Let $\tilde e$ be an additional cut-off function such that

1. $\tilde e$ is smooth and $0\leq\tilde e\leq1$.

2. $\tilde e$ depends only on $\tilde\Phi$; $\tilde e$ equals one in a neighborhood of the point $c$ and $\tilde e$ equals zero in a neighborhood of the boundary of $M_p$.

3. All partial derivatives of $\tilde e$ satisfy 
\bee\label{eq:partial_estimates'}
|\BD^\bs_{\tilde\Phi} \tilde e|\ll C_{|\bs|}.  
\ene

It is not hard to see that such function exists. Then we split \eqref{twod1} into two integrals. The integral with additional factor $(1-\tilde e)$ can be estimated as in Case 1. For the integral with additional factor $\tilde e$ we notice that now all sums $\sum_{j=1}^{K+1} b_j^t \sin\Phi_j$ are uniformly separated away from zero by some constant which depends on $\tilde k$ only. Thus, we can use geometric progression for denominators: 
$$
(l_t+\rho\sum_{j=1}^{K+1} b_j^t \sin\Phi_j)^{-1}=\sum\limits_{s=0}^\infty\rho^{-s-1}\frac{l_t^s}{(\sum_{j=1}^{K+1} b_j^t \sin\Phi_j)^{s+1}}.
$$
This leads us to the standard stationary phase integrals. As a result, we obtain the following asymptotics (see \cite{Ho1}, Theorem 7.7.5):
\begin{equation}\label{asymp}
\bes
&e^{i\rho''|\bx_0-\by_0|}\sum\limits_{p=0}^{\infty}c_+(p;n)(\rho'')^{-p-(d-1)/2}+e^{-i\rho''|\bx_0-\by_0|}\sum\limits_{p=0}^{\infty}c_-(p;n)(\rho'')^{-p-(d-1)/2}\\
-&
e^{i\rho'|\bx_0-\by_0|}\sum\limits_{p=0}^{\infty}c_+(p;n)(\rho')^{-p-(d-1)/2}-e^{-i\rho'|\bx_0-\by_0|}\sum\limits_{p=0}^{\infty}c_-(p;n)(\rho')^{-p-(d-1)/2}.
\end{split}
\end{equation}
For coefficients $c_\pm(p;n)$ we have the estimate
\begin{equation}\label{cpm}
|c_\pm(p;n)|\ll \rho_n^{p/2}
\end{equation}
which follows from the bound on $l_t$. As before, these estimates are uniform in $X$.

This covers the case of integrals over the non-resonant regions. The rest of the integration takes place over the resonant zones. Strictly speaking, those resonant zones are twice the width of the resonant zones we had before (since the support of $e_2$ can spread to the region of width $2\rho_n^{\al_1}$, whereas the width of our zones was $\rho_n^{\al_1}$). 
Nevertheless, we can extend the coordinates $X$ to these wider resonant zones. 

Consider now the contribution to \eqref{quadratic2} from a resonant partition function; by employing a further partition of unity, we can assume that the integration takes place inside one resonant region $\Bxi(\GV)_p$ and therefore the integral has the following form (we assume the 
$\Phi$ coordinates to be fixed and consider only the integral in $X$  
variables):
\begin{equation}\label{twod2}
e^{\pm i\langle\bxi,(\bx_0-\by_0)_{\GV^\perp}\rangle}\int_{\Om(\GV)} \frac{e(\rho,X)F(X) \,dX}{\prod_{t=1}^T (l_t(X)+\rho \tilde b^t)^{k_t}} e^{\pm i\langle X,(\bx_0-\by_0)_\GV\rangle}.
\end{equation}
Here, $F(X)$ is a smooth function with bounded derivatives (the bound can depend only on $\tilde k$) and $e (\rho,X)$ is a smooth cut-off  that satisfies the estimates
\bee\label{eq:partial_estimates3}
|\BD^\bs_{X}e|\ll\rho_n^{-|\bs|\beta}.  
\ene
Finally (see above), $l_t (X):=\lu X,(\bth_t)_{\GV}\ru+L_{m+1}b^t>>\rho_n^\beta$. Again, each integration by parts (integrating the exponential and differentiating other factors along $(\bx_0-\by_0)_\GV$) improves the estimate by at least $\rho_n^{-\beta}$, and thus the corresponding contribution is $O(\rho_n^{-M})$.

Now, integrating \eqref{asymp} over $X$ and $\rho$ and repeating the arguments from the diagonal case (see in particular \eqref{quadratic2} and \eqref{eq:form2}) we obtain the following asymptotic expansion for $e([-\infty,\lambda];H;\bx_0,\by_0)$ (here we also use that $e([-\infty,\lambda];H;\bx_0,\by_0)$ is real-valued):
\begin{equation}\label{I_nasymp1off}
\begin{split}&
e([-\infty,\lambda];H;\bx_0,\by_0)=
\cos(\lambda^{1/2}|\bx_0-\by_0|)\sum\limits_{p=-d+1}^{4M}\hat a(p;n)\lambda^{-p/2-(d-1)/4}+\cr & \sin(\lambda^{1/2}|\bx_0-\by_0|)\sum\limits_{p=-d+1}^{4M}\check a(p;n)\lambda^{-p/2-(d-1)/4}+A_0(n)+O(\rho_n^{-M}),\ \ \ \lambda^{1/2}\in I_n.
\end{split}
\end{equation}
Here, the real-valued coefficients $\hat a(p;n)=\hat a(p;n)(\bx_0,\by_0)$,  $\check a(p;n)=\check a(p;n)(\bx_0,\by_0)$ satisfy
\bee\label{I_nasymp2off}
|\hat a(p;n)|+|\check a(p;n)|\ll\rho_n^{p/2+d/2},
\end{equation}
and the $O$-term does not depend on $n$ (though it can depend on $M$).
The constant term $A_0(n)=A_0(n;\bx_0,\by_0)$ plays the same role as the constant term $c_n$ (see \eqref{eq:cn}) in the diagonal case:
\bee\label{eq:A0}
\bes
&A_0(n):=e([-\infty,\rho_n^2];H_2;U;\bx_0,\by_0)-\cos(\rho_n|\bx_0-\by_0|)\sum\limits_{p=-d+1}^{4M}\hat a(p;n)\rho_n^{-p-(d-1)/2}
\cr & -\sin(\rho_n|\bx_0-\by_0|)\sum\limits_{p=-d+1}^{4M}\check a(p;n)\rho_n^{-p-(d-1)/2}.
\end{split}
\ene
Obviously, we have 
\bee
|A_0(n)|\ll\rho_n^d.
\ene

 This proves Lemma~\ref{main_lemoff}.


%

%

\end{document}